\documentclass[journal]{IEEEtran}

\usepackage{amsmath,amssymb,amsfonts,mathtools,bm,mathrsfs}
\usepackage{enumitem}
\usepackage{microtype}
\usepackage{xcolor}
\usepackage{array}
\usepackage{booktabs}
\usepackage{hyperref}
\usepackage{braket}

\newcommand{\R}{\mathbb{R}}
\newcommand{\C}{\mathbb{C}}

\providecommand{\ii}{\mathrm{i}}

\newcommand{\Hermzero}{\operatorname{Herm}_0}
\newcommand{\Ad}{\operatorname{Ad}}
\newcommand{\Span}{\operatorname{span}}
\newcommand{\Tr}{\operatorname{Tr}}
\newcommand{\Range}{\operatorname{Range}}
\newcommand{\Ker}{\operatorname{Ker}}

\newcommand{\alg}{\operatorname{alg}}
\newcommand{\rank}{\operatorname{Rank}}

\newtheorem{theorem}{Theorem}

\newtheorem{corollary}{Corollary}

\newtheorem{definition}{Definition}
\newtheorem{example}{Example}
\newtheorem{exercise}{Exercise}
\newtheorem{lemma}{Lemma}

\newtheorem{problem}{Problem}
\newtheorem{proposition}{Proposition}
\newtheorem{remark}{Remark}

\newcommand{\bthm}{\begin{theorem}}
	\newcommand{\ethm}{\end{theorem}}
\newcommand{\blem}{\begin{lemma}}
	\newcommand{\elem}{\end{lemma}}
\newcommand{\bex}{\begin{example}}
	\newcommand{\eex}{\end{example}}
\newcommand{\beg}{\begin{exercise}}
	\newcommand{\eeg}{\end{exercise}}
\newcommand{\bprop}{\begin{proposition}}
	\newcommand{\eprop}{\end{proposition}}
\newcommand{\bplm}{\begin{problem}}
	\newcommand{\eplm}{\end{problem}}
\newcommand{\bmrk}{\begin{remark}}
	\newcommand{\emrk}{\end{remark}}
\newcommand{\bdfn}{\begin{definition}}
	\newcommand{\edfn}{\end{definition}}
\newcommand{\bcor}{\begin{corollary}}
	\newcommand{\ecor}{\end{corollary}}

\newcommand{\beq}{\begin{equation}}
	\newcommand{\eeq}{\end{equation}}
\newcommand{\beqm}{\begin{equation*}}
	\newcommand{\eeqm}{\end{equation*}}
\newcommand{\beqn}{\begin{eqnarray}}
	\newcommand{\eeqn}{\end{eqnarray}}
\newcommand{\beqnm}{\begin{eqnarray*}}
	\newcommand{\eeqnm}{\end{eqnarray*}}
\newcommand{\bea}{\begin{aligned}}
	\newcommand{\eea}{\end{aligned}}
\newcommand{\bs}{\begin{subequations}}
	\newcommand{\es}{\end{subequations}}

\newcommand{\bei}{\begin{itemize}}
	\newcommand{\eei}{\end{itemize}}
\newcommand{\bed}{\begin{description}}
	\newcommand{\eed}{\end{description}}
\newcommand{\bee}{\begin{enumerate}}
	\newcommand{\eee}{\end{enumerate}}
\newcommand{\bey}{\begin{array}}
	\newcommand{\eey}{\end{array}}
\newcommand{\bec}{\begin{center}}
	\newcommand{\eec}{\end{center}}

\begin{document}
	
	\title{On the Structure of Quantum Finite-Level Systems:
Decoherence-Free Subsystems, Quantum Nondemolition
Variables, and Back-Action-Evading Measurements}

	\author{Jia~Gu,
		Guofeng~Zhang,~\IEEEmembership{Senior~Member,~IEEE,}
		and~Ian~R.~Petersen,~\IEEEmembership{Fellow,~IEEE}%
		\thanks{This work was partially supported by Hong Kong Research Grant Council (RGC) under
			Grant No. 15213924 and the Australian Research Council under grant DP240101494.}%
		\thanks{J.~Gu is with the Department of Data Science and Artificial Intelligence, and G.~Zhang is with the Department of Applied Mathematics,
			The Hong Kong Polytechnic University, Hung Hom, Kowloon, Hong Kong
			(e-mail: joanna.gu@connect.polyu.hk; guofeng.zhang@polyu.edu.hk).}%
		\thanks{I.~R.~Petersen is with the School of Engineering, The Australian
			National University, Canberra, ACT 2601, Australia
			(e-mail: i.r.petersen@gmail.com).}}
	
	\markboth{IEEE Transactions on Automatic Control}%
	{Gu \MakeLowercase{\textit{et al.}}: On the Structure of Quantum Finite-Level Systems}
	
	\maketitle

	\begin{abstract}
	The purpose of this paper is to conduct a comprehensive structural analysis of quantum finite-level systems in the Heisenberg picture, where the evolution of system observables is modeled by 
		bilinear quantum stochastic differential equations (QSDEs). Firstly, real coordinate transformations that preserve commutation and anticommutation relations 
		are characterized, 
		and relative admissibility between different orthonormal Hermitian bases
		is formulated.  Secondly, strongly aligned  subspaces are
		introduced, based on which decoherence-free subsystems (DFSs) are constructed in general orthonormal Hermitian basis as well as the generalized Gell-Mann basis.  Thirdly, conditions for the existence of quantum nondemolition (QND) variables are given in terms of both linear systems theory and spectral separation by  output-signal observables. Finally,  a sufficient condition for quantum back-action-evading (BAE) measurements is established. A four-qubit nearest-neighbor Heisenberg chain with collective
		dephasing illustrates these structural properties of quantum finite-level systems.
	\end{abstract}

	
	\IEEEpeerreviewmaketitle

	\section{Introduction}
	\label{sec:introduction}
	
\IEEEPARstart{A}{quantum} finite-level system has an underlying finite-dimensional Hilbert space. Its states are positive semidefinite, unit-trace density matrices, and its observables are Hermitian matrices. This class of quantum systems includes qubits and qudits, finite spin networks, and effective multilevel descriptions of atoms, quantum dots, and superconducting artificial atoms. Finite-level systems provide basic carriers of quantum information and serve as controllable components of quantum memories, quantum processors, and quantum sensors \cite{AltafiniTicozzi2012,KnillLaflammeViola00,ClerkEtAl10}. Their usefulness depends on our ability to engineer  selected degrees of freedom, protect stored information against environmental disturbances, and extract information without severely perturbing the quantities of interest.

In quantum computing and metrology, a prominent goal is to identify and protect degrees of freedom immune to specified environmental noise. In the Heisenberg picture, a decoherence-free subsystem (DFS) comprises degrees of freedom neither driven by the environmental inputs nor detectable through the outputs  \cite{NY14, ZGPG18,ZPL20}. It need not be isolated from other system components, provided their interactions preserve these properties  \cite[Fig. 2]{ZGPG18}.   This notion is closely related to the concept of decoherence-free subspace in the Schrödinger picture \cite[Sec. III-C]{DZWW23}, which play a crucial role in protecting quantum information for quantum computation \cite{ZanardiRasetti97,LidarChuangWhaley98,KnillLaflammeViola00,TicozziViola2008,V09,DeschampsEtAl2016}.

Whereas a DFS protects internal degrees of freedom, quantum back-action-evading (BAE) measurement concerns selected input–output relations. An output field observable is BAE with respect to specified input field observables when their influence on the measured output observable vanishes. This excludes the associated back-action noise from the readout without requiring the entire system to remain undisturbed. Such measurements can enable sensing beyond the standard quantum limit (SQL) \cite{TC10,WC13,NY14,OckeloenKorppiEtAl16,ZGPG18,LOW+21,BQND24}.  Quantum BAE measurements also provide a basis for measurement-based
state preparation and feedback stabilization
\cite{ARM11,WC13,CSR+20,LD25}.

A related but distinct concept is a quantum nondemolition (QND) observable: a system observable that admits repeated measurements without measurement back-action perturbing its subsequent evolution. Specifically, in the Heisenberg picture, an observable $O$ is called a QND observable if \([O(t), O(s)]=0\) for all $t,s$.  Measurement back-action may disturb noncommuting observables, but that disturbance must not feed back into \(O\) \cite{TC10,WC13,NY14,ZGPG18,LOW+21}.

Structural decomposition provides a natural approach to visualize all of these important notions.   Indeed,  in the quantum Kalman canonical form developed in \cite{ZGPG18,ZPL20,ZLDP25,DLZ26} for linear quantum systems,  a DFS corresponds to the uncontrollable/unobservable subsystem,  QND variables corresponds to the uncontrollable/observable subsystem \cite{PDP17}, and quantum BAE measurements is characterized by a zero transfer matrix from a certain input to a certain output. Extending these results in linear quantum systems to finite-level quantum systems is challenging, mainly because a finite-level quantum system  is modeled in the  Heisenberg picture by bilinear quantum stochastic differential equations (QSDEs) where input field increments multiply system observables. Moreover, the multiplication table of the observables involves both the antisymmetric structure tensor, which determines commutators, and the symmetric structure tensor, which determines the traceless part of anticommutators \cite{Duffaut-et-al_2016}. Thus, a finite-level quantum system often has more complex dynamics in contrast to a linear one, imposing significant difficulties in system analysis and synthesis. 
	
 Structural decomposition for quantum two-level systems has been studied in \cite{ZP20}, where  DFS, QND variables and quantum BAE measurements have been identified. It should be noted that in the Schrödinger picture, 
	decoherence-free subspaces and noiseless subsystems protect quantum
	information
	\cite{ZanardiRasetti97,LidarChuangWhaley98,
		KnillLaflammeViola00,ViolaKnillLloyd00},  
	their dynamical invariance and control properties have 
	been investigated	\cite{TicozziViola2008,V09};  and   model reduction for
	open quantum finite-level dynamics  have recently been developed \cite{GrigolettoTicozzi2024,GrigolettoTaoTicozziViola2025,GrigolettoPellegriniTicozzi2025,GrigolettoTicozzi2025}; see also \cite{GZP19,GZP20,GDPD20, AMR25,HIN25,HIN26} for model reduction of quantum  filters. 

The purpose of this paper is to conduct a comprehensive study of DFSs, QND variables and quantum BAE measurements of finite-level quantum systems. The main contributions are summarized as follows.

\begin{enumerate}[label=(\roman*)]
	\item
		We characterize the real coordinate transformations (Definition \ref{def:admissible}) that
		preserve generalized Gell--Mann structure tensors and the physical realizability of finite-level quantum systems (Lemma \ref{lem:rotation}). 
		These transformations are exactly the adjoint
		representations of unitary conjugations (Theorem \ref{thm:admissible_inner}). 
		
		\item
		We introduce strongly aligned subspaces (Definition \ref{def:strong_alignment}), whose
		associated unital \(C^*\)-algebras admit an exact orthogonal
		separation into central and noncentral Hermitian
		components (Lemma \ref{rem:complex}); based on which we apply the Artin--Wedderburn  decomposition to identify a set of system observables whose evolution is decoupled from  the input and output fields (Theorem \ref{thm:decomp}).

		\item We  formulate
		relative admissibility (Definition \ref{def:rel_admissible}) between the generalized Gell-Mann basis \(\lambda\) and another orthonormal
		Hermitian basis \(\eta\) of finite-level quantum systems (Lemma \ref{lem:rotation_frame}). Then in Theorem \ref{thm:decomp_frame} we construct a DFS under the  \(\eta\)-basis, and in Corollary \ref{cor:fixed_frame_decomp} a DFS  under the   generalized Gell-Mann basis.

		\item We derive  QND variables in the sense of linear systems theory in Theorem \ref{thm:july27_qnd_G7}, and further distinguish them   in terms of
		spectral separation of the QND algebra resulting from the
		output-signal observables (Theorem \ref{thm:QND_2}).
		
		\item
		We derive a sufficient condition for BAE measurements in finite-level quantum systems  in Theorem \ref{thm:july27_G8_general}. 
	\end{enumerate}
	
	A four-qubit nearest-neighbor Heisenberg chain with
 collective dephasing illustrates  the main results. 
	
	The remainder of the paper is organized as follows.
	Section~\ref{sec:preli} presents the finite-level  quantum systems under consideration.
	Section~\ref{sec:structural decomp} develops admissible
	transformations, strong alignment, and  structural
	decomposition results. Sections \ref{sec: DFS}, \ref{sec:QND} and \ref{sec:BAE} study DFSs, QND
	observables and quantum BAE measurements, respectively.
	Section~\ref{sec:four_qubit_chain} presents a four-qubit
	illustrative example, and the final section concludes the paper.
	
\textit{Notation.} \(\ii =\sqrt{-1}\). Let \(\mathbb{R}/\mathbb{C}\) be the set of real/complex numbers.  \(I_n\) is the identity matrix of size \(n\). The subscript ``$n$'' is often omitted. Let $J$ denote the symplectic matrix  $J = \bigl[
\begin{smallmatrix}
	0 & I \\ 
	-I & 0
\end{smallmatrix}
\bigr]$.   Given a matrix \(A\), denote its \(k\)-th row by \(A_k\) or  \((A)_k\), and its transpose, entrywise conjugation, and conjugate transpose by \(A^\top,  A^\#, A^\dagger\) respectively.   $\Range(A)$ is the column space and \(\Ker(A)\) is the nullspace of $A$.     The commutator of two operators \(a,b\) is  \([a,b]=ab-ba\), while their anticommutator is \(\{a,b\}=ab+ba\).  The reduced Planck constant \(\hbar\) is set to 1.

	\section{Preliminaries}	 \label{sec:preli}
	
	
	Let  \(M_n(\mathbb{C})\)  denote the vector space of all $n\times n$ complex matrices. Denote by \(SU(n)\) the special unitary group of degree \(n\).  The Lie algebra \(\mathfrak{su}(n)\) of \(SU(n)\)    is usually written mathematically as the vector space of traceless skew-Hermitian matrices. However, in quantum physics it is often more convenient to use the  equivalent traceless Hermitian generator convention  \cite{Duffaut-et-al_2012,Duffaut-et-al_2016}.   Let \(e_1,\dots,e_n\) be the standard basis of \(\mathbb C^n\), and define  matrices \(	P_{jk}=e_je_k^\top\).	A standard traceless Hermitian basis of the Lie algebra \(\mathfrak{su}(n)\) is given by
\[
\begin{aligned}
		&u_{jk}=P_{jk}+P_{kj}, \quad 
		v_{jk}=\ii(P_{jk}-P_{kj}), \quad 1\le j<k\le n,
		\\ 
		&w_\ell
		=
		-\sqrt{\frac{2}{\ell(\ell+1)}}
		\Bigl(
		\sum_{q=1}^{\ell}P_{qq}
		-
		\ell P_{\ell+1,\ell+1}
		\Bigr), \quad  1\le \ell\le n-1.
		\end{aligned}
		\] 
	These \(s=n^2-1\) matrices form the generalized Gell-Mann basis of the vector space of traceless Hermitian matrices
	\[
	\Hermzero(n) \triangleq \{A\in M_n(\mathbb C): A=A^\dagger, \Tr(A)=0)\}.
	\]   
	Without particular order, we label these traceless Hermitian generators as \(\lambda_1,\dots,\lambda_s\), 
	which by construction satisfy
	\begin{equation} \label{eq:lambda}
		\operatorname{Tr}(\lambda_j\lambda_k)=2\delta_{jk},
	\end{equation}
	and whose commutation and anticommutation relations are
	\begin{subequations} \label{eq:CCR_antiCCR}
		\begin{align}
			[\lambda_a,\lambda_b]
			&=\;
			2\ii\sum_{c=1}^{s}f_{abc}\lambda_c, 
			\label{lambda_ccr}
			\\
			\{\lambda_a,\lambda_b\}
			&=\;
			\frac{4}{n}\delta_{ab}I_n
			+
			2\sum_{c=1}^{s}d_{abc}\lambda_c. 
			\label{lambda_anticcr}
		\end{align}
	\end{subequations}
		From \eqref{eq:CCR_antiCCR}, we have
	\begin{equation}
		\lambda_a\lambda_b
		=
		\frac{2}{n}\delta_{ab}I_n
		+
		\sum_{c=1}^{s}
		\left(
		\ii f_{abc}+d_{abc}
		\right)\lambda_c.
		\label{eq:gellmann-product}
	\end{equation}
The 3rd-order real tensor \(f_{abc}\) is completely antisymmetric, whereas \(d_{abc}\) completely symmetric:
	\begin{equation} \label{eq:sym_antisym}
		f_{abc} = -f_{bac} = f_{bca}, \quad 		d_{abc} = d_{bac} = d_{bca}.
	\end{equation}	
	More details can be found \cite[Section 3]{Duffaut-et-al_2016}.
	
	For each \(a=1,\dots,s\), we define antisymmetric matrices \(F_a\) and symmetric matrices  \(D_a\in\mathbb R^{s\times s}\) by
	\begin{equation} \label{eq:F_and_D}
		(F_a)_{bc}=f_{abc},
		\qquad
		(D_a)_{bc}=d_{abc},
	\end{equation}
	based on which define linear maps $\Theta_n^-, \Theta_n^+: \mathbb{C}^s\to\mathbb{C}^{s\times s}$ as
	\begin{subequations} \label{eq:theta-definitions}
		\begin{align}
			\Theta_n^-(\beta)
			&\triangleq 
			\begin{bmatrix}
				F_1^\top \beta&\cdots&F_s^\top \beta
			\end{bmatrix},
			\label{eq:theta-definitions_CCR}
			\\
			\Theta_n^+(\beta)
			&\triangleq
			\begin{bmatrix}
				D_1^\top \beta&\cdots&D_s^\top \beta
			\end{bmatrix}.
			\label{eq:theta-definitions_antiCCR}
		\end{align}
	\end{subequations}
	When
	$\beta$ is an $s$-dimensional row vector, then  it is understood  that $\Theta^-(\beta) \equiv \Theta^-(\beta^\top )$ and
	$\Theta^+(\beta)\equiv\Theta^+(\beta^\top )$.  
	
	\begin{lemma} \label{lemL:CCR_ACR}
		For every
		\(
		\beta\in\mathbb C^{s},
		\)
		\begin{equation} \label{eq:Theta_pm}
			\Theta_n^{-}(\beta)
			=
			\sum_{j=1}^{s}
			\beta_jF_j,
			\qquad
			\Theta_n^{+}(\beta)
			=
			\sum_{j=1}^{s}
			\beta_jD_j.
		\end{equation}
	\end{lemma}
	
	\begin{IEEEproof}
		We prove the first identity by comparing matrix entries. By definition, the $(k,\ell)$-th entry of \(\Theta_n^{-}(\beta)\) is
\[
			\bigl(\Theta_n^{-}(\beta)\bigr)_{k\ell}
			=
			\sum_{j=1}^s (F_\ell)_{jk}\beta_j
			=
			\sum_{j=1}^s f_{\ell jk}\beta_j 
			=
			\Bigl(
			\sum_{j=1}^s \beta_j F_j
			\Bigr)_{k\ell},
\]
		where the  complete antisymmetry of \(f_{ijk}\) in \eqref{eq:sym_antisym} has been used.
		Hence $
		\Theta_n^{-}(\beta)
		=
		\sum_{j=1}^s \beta_jF_j$. 
		The proof for \(\Theta_n^{+}\) is identical, using the complete
		symmetry of \(d_{ijk}\).
	\end{IEEEproof}
	
	Define the \emph{normalized}  Hilbert--Schmidt inner product 
	\begin{equation} \label{eq:nHS}
		\langle X,Y \rangle_{\rm nHS}  \triangleq  	\frac1{2}\operatorname{Tr}(X^\dagger Y), \quad X,Y\in \mathbb{C}^{n\times n}.
	\end{equation}
	Thus by \eqref{eq:lambda},  the  generalized Gell--Mann matrices \(\{\lambda_1,\dots,\lambda_s\}\) form an orthonormal basis of $\Hermzero(n)$, under the normalized Hilbert--Schmidt inner product.
	
Following Ref.~\cite{Duffaut-et-al_2016}, for an \(n\)-level open quantum system, we identify  the vector of spin observables \(X=(X_1,\dots,X_s)^\top \) with  \((\lambda_1,\dots,\lambda_s)^\top \).  Then by Lemma \ref{lemL:CCR_ACR}, Eq.~\eqref{eq:CCR_antiCCR} translates to
	\begin{subequations} \label{eq: CCR_ACR}
		\begin{align}
			[X,X^\top ]
			&\triangleq XX^\top - (XX^\top)^\top =
			2\ii\,\Theta_n^-(X),
			\label{eq:CCR}
			\\
			\{X,X^\top \}
			&\triangleq XX^\top +(XX^\top)^\top =
			\frac{4}{n}I_s
			+
			2\Theta_n^+(X).
			\label{eq:CCR_ACR}
		\end{align}
	\end{subequations}

	\begin{remark} \label{rem:op_mat}
		In \eqref{eq:CCR},  \(X\) is understood as a vector of operators.   Like \eqref{eq:theta-definitions_CCR}, 
 $
			\Theta_n^-(X) = 		\begin{bmatrix}
				F_1^\top X&\cdots&F_s^\top X
\end{bmatrix}$. 
		Hence, by Lemma \ref{lemL:CCR_ACR}, 
	 $
			\Theta_n^{-}(X)
			=
			\sum_{j=1}^{s}
			\lambda_j F_j.
	$ 
		If   \(X\) is understood as a vector of matrices as each $\lambda_j\in\mathbb{C}^{n\times n}$, then  the above $\Theta_n^-(X) $ should be understood as $
			\Theta_n^-(X) = 		\begin{bmatrix}
				(F_1^\top\otimes I_n) X&\cdots& (F_s^\top\otimes I_n) X
			\end{bmatrix}$.
		The same is true for \eqref{eq:CCR_ACR}.
	\end{remark}
	
	The evolution of the spin observables $X$ of  an \(n\)-level quantum system driven by  \(n_w\) Bosonic input fields  can be described in the Heisenberg picture by bilinear QSDEs:
	\begin{equation}  \label{sys_n_level}
		\begin{aligned}
			&dX
			=
			A_0\,dt
			+
			A X\,dt
			+
			\sum_{k=1}^{n_w} B_{1k}X\,d W_{1k}
			+
			\sum_{k=1}^{n_w} B_{2k}X\,d W_{2k},
			\\
			&\begin{bmatrix}
				d Y_1\\
				d Y_2
			\end{bmatrix}
			=
			\begin{bmatrix}
				C_1\\
				C_2
			\end{bmatrix}
			X\,dt
			+
			\begin{bmatrix}
				d W_1\\
				d W_2
			\end{bmatrix},
		\end{aligned}
	\end{equation}
	where matrices \(A_0\in\mathbb{R}^s\), \(A,B_{1k},B_{2k}\in\mathbb{R}^{s\times s}\)  and $C_1,C_2\in\mathbb{R}^{n_w\times s}$, while the  column vectors of  input operators and  output operators are
	\[
	W_j = \left[
	\begin{array}{c}
		W_{j1}\\
		\vdots\\
		W_{jn_w}
	\end{array}	
	\right], \quad 
	Y_j = \left[
	\begin{array}{c}
		Y_{j1}\\
		\vdots\\
		Y_{jn_w}
	\end{array}	
	\right],\quad j=1,2.
	\]
	Using the compact notation $
		\mathcal W
		=
		\left[\begin{smallmatrix}
			W_1\\
			W_2
	\end{smallmatrix}\right]$, $
		Y
		=\left[
				\begin{smallmatrix}
		Y_1\\
		Y_2
	\end{smallmatrix}\right]$,
	and
	\begin{equation*}
		C
		=
		\begin{bmatrix}
			C_1\\
			C_2
		\end{bmatrix}, \quad	 B_\mu
		=
		\begin{cases}
			B_{1\mu}, & 1\le \mu\le n_w,\\[2mm]
			B_{2,\mu-n_w}, & n_w<\mu\le 2n_w=m,
		\end{cases}
	\end{equation*}
the system \eqref{sys_n_level} can be rewritten as
	\begin{equation} \label{eq:irka_sys2}
		\begin{aligned}
			dX
			&=
			(A_0+AX)\,dt
			+
			\sum_{\mu=1}^m B_\mu X\,d \mathcal W_\mu,
			\\
			d Y
			&=
			CX\,dt+d \mathcal W.
		\end{aligned}
	\end{equation}
	The  It\^o matrix for the input field quadratures is
	\begin{equation} \label{eq:W}
		d \mathcal W\,d \mathcal W^\top 
		=
		\Omega\,dt,
		\qquad
		\Omega \triangleq
		\begin{bmatrix}
			I_{n_w} & \ii I_{n_w}\\
			-\ii I_{n_w} & I_{n_w}
		\end{bmatrix};
	\end{equation}
	 i.e.,  the input fields are initialized in the vacuum  state. 	
	
	The open \(n\)-level  quantum system \eqref{sys_n_level} can alternatively be parametrized by   physical parameters
	\begin{equation} \label{eq:physical_parameters}
		\theta=(\alpha,C_1,C_2),
	\end{equation}
	where \(\alpha\in\mathbb R^s\) and  \(C_1,C_2\) are  given  in \eqref{sys_n_level}.   
	With these parameters, in the $(S,L,H)$ formalism \cite{GJ09,GJ09b}, the scattering matrix $S=I_{n_w}$, the system Hamiltonian $H$ and the coupling  operator $L$ are 
	\begin{equation*}\label{eq:H and L}
		H = \alpha^\top X, \quad L 
		=
		\Lambda X, 
	\end{equation*}
	with
	\begin{equation} \label{eq:Lambda}
		\Lambda =  \frac12(C_1+\ii C_2)\in\mathbb{C}^{n_w\times s}.
	\end{equation}   
	Then the \emph{real} system  matrices in \eqref{sys_n_level} are parametrized in terms of \(\alpha, C_1, C_2\) as, \cite{Duffaut-et-al_2016}, 
	\begin{subequations}\label{eq:N-from-C}
		\begin{align}
			B_{1k}
			&=
			\Theta_n^-\bigl((C_2)_k\bigr),
			\qquad
			k=1,\dots,n_w, 
			\label{eq:B1-from-C}
			\\
			B_{2k}
			&=
			-\Theta_n^-\bigl((C_1)_k\bigr),
			\qquad
			k=1,\dots,n_w, 
			\label{eq:B2-from-C}
			\\
			A_0
			&=
			\frac{4\ii}{n}
			\sum_{k=1}^{n_w}
			\Theta_n^-(\Lambda_k^\#)\Lambda_k^\top , 
			\label{eq:A0-full}
			\\
			A
			&=
			-2\Theta_n^-(\alpha)
			+
			\sum_{k=1}^{n_w}
			\left(
			R_k-\ii Q_k
			\right),
			\label{eq:A-full}
		\end{align}
	\end{subequations}
where
	\begin{equation*}
		\begin{aligned}
			R_k
			&=
			\Theta_n^-(\Lambda_k)\Theta_n^-(\Lambda_k^\#)
			+
			\Theta_n^-(\Lambda_k^\#)\Theta_n^-(\Lambda_k),
			\\
			Q_k
			&=
			\Theta_n^-(\Lambda_k)\Theta_n^+(\Lambda_k^\#)
			-
			\Theta_n^-(\Lambda_k^\#)\Theta_n^+(\Lambda_k).
		\end{aligned}
	\end{equation*}
Finally, system matrices must satisfy
	\begin{equation}\label{eq:Ar-drift_0}
		A+A^\top 
		+
		\sum_{k=1}^{n_w}
		\left(
		B_{1k}B_{1k}^\top 
		+
		B_{2k}B_{2k}^\top 
		\right)
		=
		\frac{n}{2}\Theta_n^+(A_0).
	\end{equation}
	
It is proved in \cite{Duffaut-et-al_2016} that \eqref{eq:irka_sys2} is a valid  open \(n\)-level  quantum system if the system matrices satisfy \eqref{eq:N-from-C}-\eqref{eq:Ar-drift_0}.
	
	\section{Structural decomposition of quantum finite-level systems}	 \label{sec:structural decomp}
	
	In this section, we develop a structure-preserving decomposition theory
	for the open \(n\)-level  quantum system \eqref{eq:irka_sys2}.

	\subsection{Admissible maps}
	
	In this subsection, we establish conditions for coordinate transformations on the quantum finite-level system \eqref{eq:irka_sys2},  under which the   commutation and anticommutation relations  are preserved by the transformed system observables.
	
	\blem \label{lem:rotation}
	Given \(T\in\mathbb R^{s\times s}\), the vector of transformed observables  $\widetilde X=T^{\top}X$  have  commutation and anticommutation relations
	\beq \label{eqLjuly27_ccr}
	[\widetilde X,\widetilde X^{\top}]
	=
	2\ii \Theta_n^{-}(\widetilde X),
	\eeq
	and
	\beq \label{eqLjuly27_anti}
	\{\widetilde X,\widetilde X^{\top}\}
	=
	\frac{4}{n}I_s+2\Theta_n^{+}(\widetilde X),
	\eeq
	if and only if the transformation matrix $T$ satisfies
	\beq
	T^{\top}T=I_s,
	\label{2.1}
	\eeq
	and for every \(\beta\in\mathbb C^s\),
	\begin{subequations} \label{eq:amdissible_maps}
		\begin{align}
			\Theta_n^{-}(T^{\top}\beta)
			=&
			T^{\top}\Theta_n^{-}(\beta)T,
			\label{2.2}
			\\ 
			\Theta_n^{+}(T^{\top}\beta)
			=&
			T^{\top}\Theta_n^{+}(\beta)T.
			\label{2.3}
		\end{align}
	\end{subequations}
	\elem
	
	\begin{IEEEproof} 
		Partition $T$ as  $T = [T_1,\ldots, T_s]$. Then  $T_k^\top \beta =\sum_j T_{jk}\beta_j$. Substitution into \eqref{eq:Theta_pm} gives
		\beqm 
		\begin{aligned}
			&	\Theta_n^-(T^\top \beta) = \sum_k (T_k^\top \beta) F_k
			= \sum_j \beta_j \sum_k T_{jk} F_k. 
		\end{aligned}
		\eeqm
		On the other hand, again by \eqref{eq:Theta_pm}, 
	$
		T^{\top}\Theta_n^{-}(\beta)T = \sum_j \beta_j T^\top  F_j T$. 
		As  \(\beta\in\mathbb C^s\) is arbitrary,  Eq. \eqref{2.2} is equivalent to 
		\beq\label{eq:aug2_temp1a}
		T^\top  F_j T =  \sum_{k=1}^s T_{jk} F_k, \quad j=1,\ldots, s.
		\eeq
		Similarly, it can be easily shown that \eqref{2.3} is equivalent to 
		\beq\label{eq:aug2_temp1b}
		T^\top  D_j T =  \sum_{k=1}^s T_{jk} D_k, \quad j=1,\ldots, s.
		\eeq
		
		From \(\widetilde X=T^\top X\),  by \eqref{eq:CCR} we have  $
		[\widetilde X,\widetilde X^\top]
		=
		T^\top[X,X^\top]T
		=
		2\ii T^\top\Theta_n^-(X)T$. 
		As a result, Eq. \eqref{eqLjuly27_ccr} is equivalent to $
		T^\top\Theta_n^-(X)T = \Theta_n^-(T^\top X)$, 
		which is equivalent to
		\beq \label{eq:aug2_temp2}
		\sum_{j=1}^s \lambda_j  T^\top F_j T =  \sum_{j=1}^s \lambda_j \sum_{k=1}^s T_{jk} F_k.
		\eeq
		Because $\lambda_1,\ldots, \lambda_s$ are linear independent,  Eq. \eqref{eq:aug2_temp2} is equivalent to  \eqref{eq:aug2_temp1a}, while the latter is equivalent  to \eqref{2.2}, as has been established.  Similarly, notice that by \eqref{eq:CCR_ACR}, $	\{\widetilde X,\widetilde X^\top\}
		=
		T^\top\{X,X^\top\}T
		=
		\frac4nT^\top T+2T^\top\Theta_n^+(X)T$.
		Hence, Eq. \eqref{eqLjuly27_anti} is equivalent to  $\frac{4}{n}I_s+2\Theta_n^{+}(T^\top X)
		= 
		\frac4nT^\top T+2T^\top\Theta_n^+(X)T$, 
		which can be written as
		\beq\label{eq:aug2_temp3}
		\frac{4}{n}\left(I_s - T^\top T \right) + 2 \sum_{j=1}^s \lambda_j \Bigl(
		\sum_k T_{jk} D_k - T^\top D_jT
		\Bigr)
		=0.
		\eeq
		As $I_n, \lambda_1,\ldots, \lambda_s$ are linearly independent,  \eqref{eq:aug2_temp3} is equivalent to  \eqref{2.1} and \eqref{eq:aug2_temp1b}, while  \eqref{eq:aug2_temp1b} is equivalent to  Eq. \eqref{2.3}.
	\end{IEEEproof}
	
	\bdfn \label{def:admissible}
	A matrix \(T\in\mathbb R^{s\times s}\) is called an \emph{admissible map} if it satisfies  \eqref{2.1}--\eqref{eq:amdissible_maps}
	\edfn
	
	In this paper, an admissible map is also referred to as  an \emph{admissible coordinate transformation}. 
	
	Under the transformation $
	\widetilde X=T^\top X$,  where $T$ is an admissible coordinate transformation matrix,   the system \eqref{eq:irka_sys2}  is transformed to
	\beq \label{eq:irka_sys_T} 
	\begin{aligned}
		d\widetilde X
		&=
		(\widetilde A_0+\widetilde A\widetilde X)\,dt
		+
		\sum_{\mu=1}^{m}
		\widetilde B_\mu\widetilde X\,dW_\mu, \\
		dY&=\widetilde C\widetilde X\,dt+dW,
	\end{aligned}
	\eeq
	where $
	\widetilde A_0=T^\top A_0$, $
	\widetilde A=T^\top AT$, 
	$\widetilde B_\mu=T^\top B_\mu T$, and $
	\widetilde C=CT$.  The physical parameters \eqref{eq:physical_parameters} transform according to $
	\widetilde\alpha=T^\top\alpha$, 
	$\widetilde C_1=C_1T$, and $
	\widetilde C_2=C_2T$. 
	It can be easily verified that $
	\widetilde B_{1k}
	=
	\Theta_n^-\big((\widetilde C_2)_k\big)$, 
	$\widetilde B_{2k}
	=
	-\Theta_n^-\big((\widetilde C_1)_k\big)$, 
	and $
	\widetilde A+\widetilde A^\top
	+
	\sum_{k=1}^{n_w}
	\left(
	\widetilde B_{1k}\widetilde B_{1k}^{\top}
	+
	\widetilde B_{2k}\widetilde B_{2k}^{\top}
	\right)
	=
	\frac n2\Theta_n^+(\widetilde A_0)$.
	
	\subsection{Characterization by unitary conjugation}
	
	In this subsection, we show that an admissible coordinate transformation $T$ is induced by a unitary \(U\in SU(n)\).
	
	Let \(U\in SU(n)\) be given, and define a  matrix \(T\in\mathbb{R}^{s\times s}\) by
	\begin{equation}
		U\lambda_jU^\dagger
		=
		\sum_{k=1}^{s}
		T_{kj}\lambda_k,
		\label{eq:july27_unitary_10.1}
	\end{equation}
	where
	\begin{equation}
		T_{kj}
		=
		\frac{1}{2}
		\Tr\left(
		\lambda_kU\lambda_jU^\dagger
		\right) 
		=\langle 
		\lambda_k, U\lambda_jU^\dagger
		\rangle_{\rm nHS}.
		\label{eq:july27_unitary_10.2}
	\end{equation}
	
	Since unitary conjugation $U\bullet U^\dagger$ preserves Hermiticity and trace, it induces
	the real adjoint representation
	\beq \label{eq: Ad_map_2sep}
	\Ad:SU(n)\longrightarrow
	\mathrm{GL}_{\mathbb R}\bigl(\Hermzero(n)\bigr),
	\quad
	\Ad_U(A)=UAU^\dagger,
	\eeq
	where $\mathrm{GL}_{\mathbb R}\bigl(\Hermzero(n)\bigr)$ denotes the general linear group of real-linear transformations acting on  $\Hermzero(n)$.  
	Under the generalized Gell--Mann basis, the coordinate
	matrix of \(\Ad_U\) is the real matrix \(T\) defined by
	\eqref{eq:july27_unitary_10.1}--\eqref{eq:july27_unitary_10.2}.
	Hence, we identify the unitary conjugation 
	\(\Ad_U\) with its coordinate representation \(T\).
	
	The following result shows that  $T=Ad_U$ for each $U\in SU(n)$ is an admissible map.
	\blem \label{lem:T}
	The matrix $T$ defined by 	\eqref{eq:july27_unitary_10.1}--\eqref{eq:july27_unitary_10.2} satisfies \eqref{2.1} and \eqref{eq:amdissible_maps}; in other words, it is admissible.
	\elem
	\begin{IEEEproof}
		As $U$ is unitary, we have
		$
			2\delta_{jk}
			=
			\Tr(\lambda_j\lambda_k)
			=
			\Tr\left[
			(U\lambda_jU^\dagger)
			(U\lambda_kU^\dagger)
			\right]
=
			2(T^{\top}T)_{jk}$, 
		where \eqref{eq:lambda} and  \eqref{eq:july27_unitary_10.1}  have been used. 
		Therefore, $T$ satisfies \eqref{2.1}. Moreover, 
		\begin{equation}
			[U AU^\dagger,U BU^\dagger]
			=
			U[A,B]U^\dagger,
			\label{eq:july27_unitary_10.6}
		\end{equation}
		and 
		\begin{equation*}
			\{U AU^\dagger,UBU^\dagger\}
			=
			U\{A,B\}U^\dagger
			\label{eq:july27_unitary_10.7}
		\end{equation*}
		for arbitrary matrices $A,B$ of compatible dimension.  On one hand, as given in the proof of Lemma \ref{lem:rotation},  Eq.~\eqref{2.2} is equivalent to \eqref{eq:aug2_temp1a}. Left- and right-multiplying both sides of \eqref{eq:aug2_temp1a} by $e_a^\top$ and $e_b$ of $\mathbb{R}^s$ respectively yield 
		\beq \label{eq:aug10_temp1a}
		\sum_{m,n=1}^s f_{jmn}T_{nb}T_{ma} = \sum_{k=1}^s f_{kab} T_{jk}.
		\eeq
		On the other hand, from \eqref{eq:july27_unitary_10.6} we have
	$
		[U \lambda_j U^\dagger,U \lambda_k U^\dagger]
		=
		U[\lambda_j ,\lambda_k ]U^\dagger,
	$
		which by \eqref{eq:july27_unitary_10.1} is equivalent to  
		\[
			\sum_{c=1}^s \sum_{a,b=1}^s f_{cab}T_{aj}T_{bk} \lambda_c = \sum_{c=1}^s \sum_{\ell=1}^s f_{\ell jk} T_{c\ell}\lambda_c.
		\]	 
		Consequently,
		\beq \label{eq:aug10_temp1b}
		\sum_{a,b=1}^s f_{cab}T_{aj}T_{bk} =  \sum_{\ell=1}^s f_{\ell jk} T_{c\ell}
		\eeq
		due to linear independence of generalized Gell-Mann matrices. 
		However, \eqref{eq:aug10_temp1a} is actually \eqref {eq:aug10_temp1b}. As this is true for all $a,b,c\in \{1,\ldots, s\}$, from \eqref{eq:july27_unitary_10.6} we get \eqref{2.2}.   The anticommutator case can be proved similarly. 
	\end{IEEEproof}
	
	In the remainder of this subsection, we aim to show that each admissible map is induced by a unitary matrix. To do this, some preparations are in order.
	
	Define the linear map
	\beq \label{eq:def_Phi}
	\Phi:\mathbb{R}^{s}\longrightarrow \Hermzero(n),
	\qquad
	\Phi(v)=v^\top X
	=\sum_{a=1}^{s}v_a\lambda_a.
	\eeq
	By \eqref{eq:gellmann-product} we have
	\beq \label{eq:isometry_4sep}
	\langle \Phi(u), \Phi(v)\rangle_{\rm nHS}=u^\top v = \langle u,v\rangle,
	\eeq
	where $\langle \cdot, \cdot \rangle$ is the usual inner product for $\mathbb{R}^s$. 
	Thus $\Phi$ is an isometry. Introduce two bilinear maps from $\mathbb{R}^s\times \mathbb{R}^s \to \mathbb{R}^s$ as
	\beq \label{eq:fd_maps}
	f(u,v)\triangleq\Theta_n^-(u)v,
	\quad
	d(u,v)\triangleq\Theta_n^+(u)v, \quad u, v \in \mathbb{R}^s.
	\eeq
	It can be readily verified that 
	\begin{subequations}	 \label{eq:Phi_uv_ccr_anticcr}
		\begin{align}
			[\Phi(u),\Phi(v)]
			&=
			-2\ii \,\Phi\bigl(f(u,v)\bigr),
			\label{eq:july25_1.1}
			\\
			\{\Phi(u),\Phi(v)\}
			&=
			\frac{4}{n}(u^\top v)I_n
			+
			2\Phi\bigl(d(u,v)\bigr),
			\label{eq:july25_1.2}
		\end{align}
	\end{subequations}	
	from which we have
	\begin{equation}
		\Phi(u)\Phi(v)
		=
		\frac{2}{n}(u^\top v)I_n
		+
		\Phi\bigl(d(u,v)\bigr)
		-\ii \Phi\bigl(f(u,v)\bigr).
		\label{eq:july25_1.3}
	\end{equation}
	
	\bmrk \label{rem:fd}
Because we will use the Skolem--Noether Theorem for \(M_n(\mathbb{C)}\) (see  e.g.,
\cite{StacksSN}) to establish Theorem \ref{thm:admissible_inner},   we extend the  real maps \(f\) and \(d\) defined in \eqref{eq:fd_maps} to complex-bilinear maps \(\mathbb{C}^s\times\mathbb{C}^s\to\mathbb{C}^s\) defined by the same formulas \(f(u,v)=\Theta^-_n(u)v\) and  \(d(u,v)=\Theta^+_n(u)v\)  for all \(u,v\in\mathbb{C}^s\). Indeed, they are well defined since the maps \(\Theta^\pm_n\) are  actually defined on \(\mathbb{C}^s\) in \eqref{eq:theta-definitions}.  
	\emrk
	
	Let $\Phi_{\mathbb C}$ denote the complex-linear extension of the isometric map $\Phi$; specifically, 
	\beq \label{eq:Phi_C}
	\Phi_{\mathbb C}(u+\ii v) \triangleq  \Phi(u)+ \ii   \Phi(v), \quad u,v\in \mathbb{R}^s.
	\eeq
	By means of \eqref{eq:july25_1.3}, it can be easily shown that   for all \(u,v\in\mathbb{C}^s\),
	\begin{equation} \label{eq:aug6_temp2}
		\Phi_\mathbb{C}(u)\Phi_\mathbb{C}(v)=\frac{2}{n}(u^\top v)I_n+\Phi_\mathbb{C}\big(d(u,v)\big)-\mathrm{i}\,\Phi_\mathbb{C}\big(f(u,v)\big).
	\end{equation}
	Similarly, given a real subspace $\mathcal V\subseteq\mathbb R^s$, define a complex space \(\mathcal{V}_\mathbb{C}\triangleq\mathcal{V}+\mathrm{i}\mathcal{V}\subseteq\mathbb{C}^s\). Then we use the notation
 $
	\Phi_{\mathbb C}(\mathcal V_{\mathbb C}) = \Phi(\mathcal V) + \ii  \Phi(\mathcal V)
$
	as the complex extension of the real vector space $\Phi(\mathcal V)$.
	
	\begin{theorem}
		\label{thm:admissible_inner}
	A matrix \(T\in\mathbb R^{s\times s}\) is admissible if and only if there exists
		\(U\in SU(n)\) such that \eqref{eq:july27_unitary_10.1} holds.
	\end{theorem}
	
	\begin{IEEEproof}
		The ``if'' part is Lemma~\ref{lem:T}. Conversely, suppose that
		\(T\) is admissible. From \eqref{2.1}--\eqref{eq:amdissible_maps} with
		\(\beta=Tu\), we obtain $
		\Theta_n^\pm(Tu)
		=
		T\Theta_n^\pm(u)T^\top$. 
		Consequently,
		\begin{align*}
			f(Tu,Tv)
			&=
			\Theta_n^-(Tu)Tv
			=
			T\Theta_n^-(u)v
			=
			Tf(u,v),
			\\
			d(Tu,Tv)
			&=
			\Theta_n^+(Tu)Tv
			=
			T\Theta_n^+(u)v
			=
			Td(u,v).
		\end{align*}
		In terms of the function $\Phi_{\mathbb C}$  defined in \eqref{eq:Phi_C}, define a complex-linear map $
		\Psi_T:M_n(\mathbb C)\longrightarrow M_n(\mathbb C)$ 
		by
		\beq \label{eq:def_Phi_T}
		\Psi_T\bigl(cI_n+\Phi_{\mathbb C}(z)\bigr)
		\triangleq 
		cI_n+\Phi_{\mathbb C}(Tz), \ \ c\in\mathbb C,  z\in\mathbb C^s.
		\eeq
This map is well defined  because 
		$M_n(\mathbb C)
		=
		\Span_\mathbb{C}\{I_n, \lambda_1,\ldots, \lambda_s\}$.  Clearly,   \(\Psi_T\) is unital and bijective. As \(T\) is real, 
		\beq \label{eq:Phi_T_A}
		\Psi_T(A^\dagger)=\Psi_T(A)^\dagger.
		\eeq	 
		Furthermore, by \eqref{eq:fd_maps} and \eqref{eq:aug6_temp2} we get $
		\Psi_T(AB)=\Psi_T(A)\Psi_T(B)$ for all $
		A,B\in M_n(\mathbb C)$. 
		Thus \(\Psi_T\) is a unital \( * \)-automorphism of
		\(M_n(\mathbb C)\). Then by the Skolem--Noether Theorem \cite{StacksSN}, there exists an invertible matrix $V$ such
		that 
	$
		\Psi_T(A)=VAV^{-1}$ 
		for all $
		A\in M_n(\mathbb C)$.
		By \eqref{eq:Phi_T_A}, $
		V A^\dagger V^{-1}
		=
		(VAV^{-1})^\dagger
		=
		(V^{-1})^\dagger A^\dagger V^\dagger$, which can be rewritten as $(V^\dagger V) A^\dagger = A^\dagger (V^\dagger V)$. In other words, \(V^\dagger V\) commutes with every element of
		\(M_n(\mathbb C)\). Therefore $
		V^\dagger V=\gamma I_n$ 
		for some \(\gamma>0\). Replacing \(V\) by \(\gamma^{-1/2}V\) gives a unitary matrix.
		Multiplying \(\gamma^{-1/2}V\) by a  phase $e^{\ii \theta}$ yields  \(U\in SU(n)\). Then $
		\Psi_T(A)=UAU^\dagger$. Applying this identity to \eqref{eq:def_Phi_T} and letting $c=0$ and $z=e_j$  establish
		\eqref{eq:july27_unitary_10.1}.
	\end{IEEEproof}
	
	\subsection{Strong alignment} \label{subsec: alignment}
	
	In this subsection, we introduce strongly aligned subspaces
	\(\mathcal S\subseteq\mathbb R^s\).  A strongly aligned subspace is
	accompanied by another subspace
	\(\mathcal Q\subseteq\mathcal S^\perp\). 
	The  spaces of observables  \(\Phi(\mathcal S)\) and \(\Phi(\mathcal Q)\), corresponding to \(\mathcal S\) and \(\mathcal Q\) respectively, will be
	identified with the noncentral and traceless central observable parts,
	respectively, of a unital \(C^*\)-algebra, \cite{MA25,DeschampsEtAl2016}.

	\begin{definition}\label{def:strong_alignment}
		A real subspace $\mathcal S\subseteq\mathbb R^s$ is called
		\emph{strongly aligned} if there exists another real subspace $
		\mathcal Q\subseteq\mathcal S^\perp$ 
		such that
		\begin{subequations}\label{eq:f_d_Q_S}
			\begin{align}
				f(\mathcal S,\mathcal S)
				&\subseteq\mathcal S,
				\label{eq:july25_2.1}
				\\
				d(\mathcal S,\mathcal S)
				&\subseteq\mathcal S\oplus\mathcal Q,
				\label{eq:july25_2.2}
				\\
				f\bigl(\mathcal Q,\mathcal S\oplus\mathcal Q\bigr)
				&=\{0\},
				\label{eq:july25_2.3}
				\\
				d(\mathcal Q,\mathcal S)
				&\subseteq\mathcal S,
				\qquad
				d(\mathcal Q,\mathcal Q)
				\subseteq\mathcal Q,
				\label{eq:july25_2.4}
			\end{align}
		\end{subequations}		
		and the subspace of $M_n(\mathbb{C})$ defined as
		\begin{equation}
			\mathfrak A
			\triangleq
			\mathbb C I_n
			\oplus
			\Phi_{\mathbb C}(\mathcal Q_{\mathbb{C}})
			\oplus
			\Phi_{\mathbb C}(\mathcal S_{\mathbb{C}})
			\label{eq:july25_2.5}
		\end{equation}
		satisfies
		\begin{subequations} \label{eq:Phi_QS}
			\begin{align}
				\Phi(\mathcal S)
				&=
				\mathfrak A_{\mathrm{sa}}
				\cap
				(Z(\mathfrak A)_{\mathrm{sa}})^\perp.
				\label{eq:july25_2.7}
				\\
				\Phi(\mathcal Q)
				&=
				Z(\mathfrak A)_{\mathrm{sa}}
				\cap
				\{I_n\}^{\perp}.
				\label{eq:july25_2.6}
			\end{align}
		\end{subequations}			
		Here, $Z(\mathfrak A)$ is the center of $\mathfrak A$, defined as
		\beq \label{eq:Z center}
		Z(\mathfrak A)
		\triangleq
		\{Z\in\mathfrak A:ZA=AZ,\ \forall A\in\mathfrak A\},
		\eeq
		the subscript
		``$\mathrm{sa}$'' denotes the  part of observables, e.g.,  
		\beqm
		\begin{aligned}
			\mathfrak A_{\mathrm{sa}} &= \{Y\in\mathfrak A: Y=Y^\dagger \}, \\
			Z(\mathfrak A)_{\mathrm{sa}} &= \{Y\in Z(\mathfrak A): Y=Y^\dagger\},
		\end{aligned}
		\eeqm
		and  orthogonality in \eqref{eq:Phi_QS} is with
		respect to the normalized Hilbert--Schmidt inner product \eqref{eq:nHS}.
	\end{definition}

Some algebraic  properties of the subspace  \(\mathfrak{A}\) defined in \eqref{eq:july25_2.5} are given by the following two lemmas.
	
	\blem\label{rem:complex}
	The decomposition of \(\mathfrak{A}\) in \eqref{eq:july25_2.5} is  indeed a direct sum.
	\elem 
	\begin{IEEEproof}	
		By definition, \(\langle I_n,\Phi_\mathbb{C}(w)\rangle_{\mathrm{nHS}}=\tfrac12\operatorname{Tr}\Phi_\mathbb{C}(w)=0\) for all $w\in\mathbb{C}^s$.  Moreover,  \(\mathcal{Q}\perp\mathcal{S}\) implies \(p^\dagger w=0\) for all \(p\in\mathcal{Q}_\mathbb{C}\) and  \(w\in\mathcal{S}_\mathbb{C}\). Thus, the three summands  in \eqref{eq:july25_2.5}   are mutually orthogonal under the normalized Hermitian inner product \eqref{eq:nHS}, and the result follows.
	\end{IEEEproof}

	\begin{lemma}\label{lem:* algebra}
		The space $\mathfrak A$ defined in \eqref{eq:july25_2.5}
		is a finite-dimensional unital  $C^*$-algebra.
	\end{lemma}
	
	\begin{IEEEproof}
		Firstly, $I_n \in \mathfrak{A}$,  and hence  $	\mathfrak A$ is unital.  Secondly,  as \((cI_n)^\dagger= c^* I_n\), and by \eqref{eq:Phi_C} and \eqref{eq:def_Phi},  \(\Phi_\mathbb{C}(w)^\dagger=\Phi_\mathbb{C}(w^\dagger)\) for all $w\in\mathbb{C}^s$, we have \(\mathfrak{A}^\dagger=\mathfrak{A}\). Finally, we show multiplicative closure. Products involving \(I_n\) are trivial.  If \(\mathcal{U},\mathcal{V},\mathcal{W}\subseteq\mathbb{R}^s\) are real subspaces with \(f(\mathcal{U},\mathcal{V})\subseteq\mathcal{W}\), then by the complex-bilinear expansion of \(f(u_x+\mathrm{i}u_y,\,v_x+\mathrm{i}v_y)\) (see Remark \ref{rem:fd}), simple algebraic manipulations give \(f(\mathcal{U}_\mathbb{C},\mathcal{V}_\mathbb{C})\subseteq\mathcal{W}_\mathbb{C}\); similarly for the map \(d\). Hence  \eqref{eq:f_d_Q_S} also holds with \(\mathcal{S},\mathcal{Q}\) replaced by \(\mathcal{S}_\mathbb{C},\mathcal{Q}_\mathbb{C}\), provided 	\(\mathcal Q\perp\mathcal S\)  as given in Definition \ref{def:strong_alignment}.  Consequently,  
		$
		\Phi_{\mathbb C}(\mathcal S_{\mathbb C})
		\Phi_{\mathbb C}(\mathcal S_{\mathbb C})
		\subseteq\mathfrak A$,  $
		\Phi_{\mathbb C}(\mathcal Q_{\mathbb C})
		\Phi_{\mathbb C}(\mathcal S_{\mathbb C})
		\subseteq\mathfrak A$, 
		and $
		\Phi_{\mathbb C}(\mathcal Q_{\mathbb C})
		\Phi_{\mathbb C}(\mathcal Q_{\mathbb C})
		\subseteq\mathfrak A$ 
		follow from \eqref{eq:aug6_temp2}. 
		Finally, equip $\mathfrak A$ with the matrix operator norm $
		\|A\| = 
		\sigma_{\rm max}(A)$,
 where  \(\sigma_{\rm max}(A)\) denotes the largest singular value of \(A\).  Clearly, \(\|A^\dagger A\| =\|A\|^2\). 
		Thus $\mathfrak A$ is a finite-dimensional unital $C^*$-algebra.
	\end{IEEEproof}	
	
	\bmrk	 
	Clearly, \([\Phi(p),I_n]=0\) for all $p \in \mathbb{R}^s$.  Moreover,  for \(p\in\mathcal{Q}\), by \eqref{eq:july25_1.1} and \eqref{eq:july25_2.3} we have  \([\Phi(p),\Phi(q)]=-2\mathrm{i}\Phi(f(p,q))=0\) for every \(q\in\mathcal{S}\oplus\mathcal{Q}\).  As a result, by the form of   \(\mathfrak{A}\)  defined in \eqref{eq:july25_2.5}, we get \(\Phi(p)\in Z(\mathfrak{A})\). However, as \(\Phi(p) \in \Hermzero(n)\), we have \(\Phi(\mathcal{Q})\subseteq Z(\mathfrak{A})_{\mathrm{sa}}\cap\{I_n\}^\perp\). In this sense,  \eqref{eq:july25_2.6} is the  maximality assumption that this inclusion is actually an equality. Moreover, once \eqref{eq:july25_2.6} holds, \(\Phi(\mathcal{S})\subseteq\mathfrak{A}_{\mathrm{sa}}\cap(Z(\mathfrak{A})_{\mathrm{sa}})^\perp\) naturally holds (for \(w\in\mathcal{S}\): \(\langle\Phi(w),cI_n+\Phi(p)\rangle_{\mathrm{nHS}}=0+w^\top p=0\)), therefore \eqref{eq:july25_2.7} is again a maximality assumption.  Such a ``maximality'' feature  alludes to  the word ``strongly'' in Definition \ref{def:strong_alignment}.
	\emrk
	
	\subsection{Artin--Wedderburn decomposition}
	
	In this subsection, by applying the Artin--Wedderburn decomposition to the  unital \( C^* \)-subalgebra
	\(\mathfrak A\subseteq M_n(\mathbb C)\) studied in the preceding subsection,  we present an interpretation of  a strongly aligned subspace $\mathcal S$ and the associated  subspace
	\(\mathcal Q\subseteq\mathcal S^\perp\).

	By Lemma~\ref{lem:* algebra}, $\mathfrak A$ is a
	finite-dimensional unital $C^*$-algebra.
	The finite-dimensional Artin--Wedderburn Theorem
	\cite{MA25} therefore applies. Specifically,  there exist positive
	integers \(r_k,m_k\) and a unitary \(U\) such that
	\beq \label{eq:AW data}
	\sum_{k=1}^{K}r_km_k=n,
	\eeq
	and
	\begin{equation}
		\widetilde{\mathfrak A} \triangleq U\mathfrak A U^\dagger
		=
		\bigoplus_{k=1}^{K}
		\left(
		M_{r_k}(\mathbb C)\otimes I_{m_k}
		\right).
		\label{eq:july25_4.1}
	\end{equation}
	Moreover, as multiplying $U$ by a phase $e^{\imath \theta}$ does not change its adjoint action in \eqref{eq:july25_4.1},  $U$ can always be chosen to satisfy $U\in SU(n)$.
	
	By the definition of  ``center'' in \eqref{eq:Z center}, the center of  $\widetilde{\mathfrak A} $  is
	\begin{equation}
		Z(\widetilde{\mathfrak A})
		=
		\bigoplus_{k=1}^{K}
		\left(\mathbb C
		I_{r_k}\otimes I_{m_k}\right),
		\label{eq:july25_4.2}
	\end{equation}
whose traceless Hermitian part  is
	\begin{equation}
		\mathfrak Z_0
		=
		\Bigl\{
		\bigoplus_{k=1}^{K}
		\left(z_k I_{r_k}\otimes I_{m_k}\right)
		:
		z_k\in\mathbb R,\ 
		\sum_{k=1}^{K}r_km_kz_k=0
		\Bigr\}.
		\label{eq:july25_4.3}
	\end{equation}
	Consequently, $
	\dim_{\mathbb R}\mathfrak Z_0=K-1$. 
	(The notation ``\(\dim_\mathbb{R}V\)'' denotes the dimension of \(V\) as a vector space over the field of \emph{real} numbers; i.e., the number of vectors in a basis when only their \emph{real}  linear combinations are allowed. For example,  $\dim_\mathbb{R} \mathbb{C}^n=2n$.) Clearly by definition, $\mathfrak Z_0$ is the traceless Hermitian center part of 	$\widetilde{\mathfrak A} $.  On the other hand,	the noncentral traceless Hermitian part of $\widetilde{\mathfrak A} $ is
	\begin{equation}
		\mathfrak D
		=
		\bigoplus_{k=1}^{K}
		\left(
		\Hermzero(r_k)\otimes I_{m_k}
		\right),
		\label{eq:july25_4.4}
	\end{equation}
	with dimension  $
	\dim_{\mathbb R}\mathfrak D
	=
	\sum_{k=1}^{K}(r_k^2-1)$. 
	
	\blem\label{lem:D_intersection}
	$\mathfrak D$ defined in \eqref{eq:july25_4.4} can be written as
	\beq \label{D_intersection}
	\mathfrak D
	=
	\widetilde{\mathfrak A}_{\mathrm{sa}}
	\cap
	(Z(\widetilde{\mathfrak A} )_{\mathrm{sa}})^\perp.
	\eeq
	\elem	
	\begin{IEEEproof}
		Clearly, by the definition of $Z(\widetilde{\mathfrak A})$) in \eqref{eq:july25_4.2}, its Hermitian part is 
		\beq \label{eq:Z_A_sa}
		Z(\widetilde{\mathfrak A})_{\mathrm{sa}} = 	\mathbb RI_n
		\oplus
		\mathfrak Z_0.
		\eeq  
		By the definition of $\widetilde{\mathfrak A} $  in \eqref{eq:july25_4.1}, every Hermitian element of $\widetilde{\mathfrak A} $ has the
		form
		$
		H
		=
		\bigoplus_{k=1}^{K}
		(H_k\otimes I_{m_k}),
		$
		where the Hermitian matrix $H_k = H_k^\dagger \in M_{r_k}(\mathbb{C})$ has the unique decomposition
		\[
		H_k
		=
		\frac{\Tr(H_k)}{r_k}I_{r_k}
		+
		H_k^0,
		\qquad
		H_k^0 \in \Hermzero(r_k).
		\]
		Clearly, the first part contributes to  $Z(\widetilde{\mathfrak A})_{\mathrm{sa}}$ in \eqref{eq:Z_A_sa}, while the second part belongs to
		$\mathfrak D$. Hence, the  Hermitian part of $\widetilde{\mathfrak A}$ is
		\beq \label{eq: A_can_sa}
		\widetilde{\mathfrak A}_{\mathrm{sa}}
		=
		\mathbb RI_n
		\oplus
		\mathfrak Z_0
		\oplus
		\mathfrak D.
		\eeq	
		We have $\mathfrak D \subseteq \widetilde{\mathfrak A}_{\mathrm{sa}}$. 
		For each \(Z=\bigoplus_k z_kI_{r_k}\otimes I_{m_k}\in\mathfrak{Z}_0\) and $D_0=\bigoplus_{k=1}^K(H_k^0\otimes I_{m_k})$, where $H_k^0\in\Hermzero(r_k)$, it is easy to show that
		$
		\langle I_n,Z\rangle_{\mathrm{nHS}}=\tfrac12\textstyle\sum_k r_km_kz_k=0$, 
		$\langle I_n,D_0\rangle_{\mathrm{nHS}}=\tfrac12\textstyle\sum_k m_k\operatorname{Tr}(H^0_k)=0$, and 
		$\langle Z,D_0\rangle_{\mathrm{nHS}}=\tfrac12\textstyle\sum_k z_km_k\operatorname{Tr}(H^0_k)=0$. 
		Now take \(Y\in\widetilde{\mathfrak{A}}_{\mathrm{sa}}\) with the unique decomposition \(Y=cI_n+Z+D_0\) as given in  \eqref{eq: A_can_sa}. By the above orthogonality, \(Y\perp Z(\widetilde{\mathfrak A} )_{\mathrm{sa}}\) if and only if \(c=0,Z=0\), i.e., if and only if  \(Y\in\mathfrak{D}\). The proof is  completed.
	\end{IEEEproof}	
	
	Recall the adjoint map $\mathrm{Ad}_U$ defined in \eqref{eq: Ad_map_2sep}. We can easily  extend it to be a map on $M_n(\mathbb{C})$ as \(\mathrm{Ad}_U(A) = U A U^\dagger\) for all $A\in M_n(\mathbb{C})$.   The following result can be easily verified, and thus its proof is omitted.
		\begin{lemma} \label{lem:adjoint map}
		Let \(U\in SU(n)\) be given. Then \(\operatorname{Ad}_U(A)=UAU^\dagger\)
		\begin{description}
		 \item[(i)] fixes \(I_n\); i.e., $\operatorname{Ad}_U(I_n) = I_n$;
			\item[(ii)] maps \(\mathfrak A\) onto \(\widetilde{\mathfrak A}\);
			\item[(iii)] maps \(\mathfrak A_{\mathrm{sa}}\) onto \(\widetilde{\mathfrak A}_{\mathrm{sa}}\);
			\item[(iv)] maps \(Z(\mathfrak A)\) onto \(Z(\widetilde{\mathfrak A})\);
			\item[(v)] is an isometry for the normalized Hilbert–Schmidt inner product \eqref{eq:nHS}.
		\end{description}
		Consequently, for subspaces \(V,W\subseteq M_n(\mathbb C)\),
		\beqm \label{eq:Ad_VW}
		\begin{aligned}
			\operatorname{Ad}_U(V\cap W)
			&=\operatorname{Ad}_U(V)\cap\operatorname{Ad}_U(W),
			\\
			\operatorname{Ad}_U(V^\perp)
			&=\operatorname{Ad}_U(V)^\perp.
		\end{aligned}
		\eeqm
	\end{lemma}	

The following result is a consequence of  Lemmas \ref{lem:D_intersection} and \ref{lem:adjoint map}.
	
	\begin{lemma} \label{lem: Z_0 and D}
		$\Phi(\mathcal S)$  and  $\Phi(\mathcal Q)$ in \eqref{eq:Phi_QS} are  mapped onto  $\mathfrak Z_0$ and $\mathfrak D$ in \eqref{eq: A_can_sa} respectively  by $\mathrm{Ad}_U$; i.e.,
		\begin{subequations}	\label{eq: Z_0 and D}
			\begin{align}
				U\Phi(\mathcal Q)U^\dagger
				&=
				\mathfrak Z_0,
				\label{eq:july25_4.6}
				\\
				U\Phi(\mathcal S)U^\dagger
				&=
				\mathfrak D.
				\label{eq:july25_4.7}
			\end{align}
		\end{subequations}	
	\end{lemma}		
	
	\begin{IEEEproof}
		First, we derive \eqref{eq:july25_4.6}. Applying \(\mathrm{Ad}_U\) to \eqref{eq:july25_2.6} and using Lemma \ref{lem:adjoint map},
		\beq
			U\Phi(\mathcal{Q})U^\dagger
			=\mathrm{Ad}_U\big(Z(\mathfrak{A})_{\mathrm{sa}}\cap\{I_n\}^\perp\big)
			=Z(\widetilde{\mathfrak{A}})_{\mathrm{sa}}\cap\{I_n\}^\perp .
			\label{eq:temp1_22aug}
		\eeq
		By \eqref{eq:july25_4.2}, a self-adjoint element of \(Z(\widetilde{\mathfrak{A}})\) has the form \(\bigoplus_k z_kI_{r_k}\otimes I_{m_k}\) with \(z_k\in\mathbb{R}\), and it is orthogonal to \(I_n\) if and only if \(\sum_k r_km_kz_k=0\). However, by \eqref{eq:july25_4.3} this set is exactly \(\mathfrak{Z}_0\). Hence by \eqref{eq:temp1_22aug},  we get  \eqref{eq:july25_4.6}. Next, we derive \eqref{eq:july25_4.7}. Applying \(\mathrm{Ad}_U\) to \eqref{eq:july25_2.7} and using Lemma \ref{lem:adjoint map},  $U\Phi(\mathcal{S})U^\dagger
		=\mathrm{Ad}_U\big(\mathfrak{A}_{\mathrm{sa}}\cap(Z(\mathfrak{A})_{\mathrm{sa}})^\perp\big)
		=\widetilde{\mathfrak{A}}_{\mathrm{sa}}\cap\big(Z(\widetilde{\mathfrak{A}})_{\mathrm{sa}}\big)^\perp=\mathfrak{D}$, 
		where the last equality is \eqref{D_intersection}.
	\end{IEEEproof}		
	
	Lemma \ref{lem: Z_0 and D} provides an  interpretation of the two coordinate subspaces \(\mathcal S, \mathcal Q\)  determined by strong alignment.  Specifically, the associated subspace
	$\Phi(\mathcal Q)$ is mapped onto the traceless Hermitian center
	$\mathfrak Z_0$, while $\Phi(\mathcal S)$ is mapped onto the 
	noncentral traceless Hermitian part $\mathfrak D$.  In this sense, the
	Artin--Wedderburn decomposition does more than merely block-diagonalize $\mathfrak A$ as shown in \eqref{eq:july25_4.1}: it separates its central 
	Hermitian matrices $\mathfrak Z_0$ (namely, Hermitian matrices commuting with each element in \(\widetilde{\mathfrak A}\)) from its noncommutative  
	observables $\mathfrak{D}$, where the latter constitute a set of observables whose evolution is decoupled from the environment  in the
	structural-decomposition result to be presented in the next subsection. 
	
	\subsection{Structural decomposition}
	\label{subsec:basis_independent_decomp}

In this subsection, we derive conditions under which there exists a set of observables of the  finite-level system \eqref{eq:irka_sys2}, whose evolution is decoupled from the input and output fields.
	
	\begin{theorem}
		\label{thm:decomp}
		Let $\mathcal S\subseteq\mathbb R^s$ be strongly aligned  and let $P_{\mathcal S}$ be its
		orthogonal projector.  Assume that
		\begin{subequations}\label{eq:zero_pattern}
			\begin{align}
				[A,P_{\mathcal S}]&=0,
				\label{eq:july25_5.4}\\
				B_\mu P_{\mathcal S}&=0,
				\qquad \mu=1,\ldots,m,
				\label{eq:july25_5.5}\\
				CP_{\mathcal S}&=0,
				\label{eq:july25_5.6}\\
				P_{\mathcal S}A_0&=0.
				\label{eq:july25_5.7}
			\end{align}
		\end{subequations}
Define observables
			\beq \label{eq: X_X_perp}
			X_{\mathcal S}\triangleq P_{\mathcal S}X,
			\qquad
			X_\perp\triangleq(I_s-P_{\mathcal S})X.
			\eeq
			Then system~\eqref{eq:irka_sys2} can be rewritten as 
			\begin{subequations}\label{eq:proj_decomp}
				\begin{align}
					dX_{\mathcal S}
					&=
					AX_{\mathcal S}\,dt,
					\label{eq:proj_S}\\
					dX_\perp
					&=
					\bigl(A_0+AX_\perp\bigr)\,dt
					+
					\sum_{\mu=1}^{m}
					B_\mu X_\perp\,dW_\mu,
					\label{eq:proj_perp}\\
					dY
					&=
					CX_\perp\,dt+dW.
					\label{eq:proj_out}
				\end{align}
			\end{subequations}
			Hence,  the  vector of observables   \(X_{\mathcal S}\) evolves
			autonomously, decoupled from both the input and output fields. 
	\end{theorem}
	
	\begin{IEEEproof}
	By \eqref{eq:july25_5.4}, the space spanned by
		$\mathcal S$ and that spanned by $\mathcal S^\perp$ are invariant under the transformation of the matrix $A$.  Moreover,
		by \eqref{eq:B1-from-C}--\eqref{eq:B2-from-C},
		$B_\mu^\top=-B_\mu$, and hence by \eqref{eq:july25_5.5},
		\beq \label{eq: P_S B}
		P_{\mathcal S}B_\mu
		=
		-(B_\mu P_{\mathcal S})^\top
		=
		0, \qquad \mu=1,\ldots,m.
		\eeq
		Left-multiplying \eqref{eq:irka_sys2} by $P_{\mathcal S}$ and using \eqref{eq:july25_5.4}, \eqref{eq:july25_5.7}, and \eqref{eq: P_S B} we get
		\eqref{eq:proj_S}.  Similarly, left-multiplying   \eqref{eq:irka_sys2} by $I_s-P_{\mathcal S}$ and  using \eqref{eq:july25_5.4}, \eqref{eq:july25_5.7}, and \eqref{eq: X_X_perp},  we get
		\eqref{eq:proj_perp}.  Finally,  \eqref{eq:july25_5.6} and \eqref{eq: X_X_perp} imply that
		$CX=C(I_s-P_{\mathcal S})X=CX_\perp$, which establishes \eqref{eq:proj_out}.
\end{IEEEproof}

\section{DFSs} \label{sec: DFS}
	
Theorem \ref{thm:decomp} identifies a set of the system observables whose evolution is not influenced by the environment.  Unfortunately, as evidenced by \eqref{eq:proj_decomp}, it does not give an explicit decomposition  of system matrices;  this is the purpose of this section.
	
	\subsection{Block-adapted bases and relative admissibility}
	\label{sec:frames}
	
In general, a coordinate transformation  does not  preserve the generalized Gell--Mann tensors \((f,d)\); instead,  it requires these structure tensors  to be transformed together with the basis. The purpose of this subsection is to construct a block-adapted  decomposition such that the leading block is a DFS under the transformed  structure tensors.

	\begin{definition}\label{def:frame}
		Let 
		$\eta=(\eta_1,\dots,\eta_s)$ be an \emph{ordered} orthonormal basis of $\Hermzero(n)$ with respect to the normalized
		Hilbert--Schmidt inner product \eqref{eq:nHS}.  Define its structure tensors
		\begin{equation}\label{eq:frame_tensors_def}
			f^\eta_{abc}
			\triangleq
			\frac{1}{4\ii}\Tr\bigl([\eta_a,\eta_b]\eta_c\bigr),
			\qquad
			d^\eta_{abc}
			\triangleq
			\frac14\Tr\bigl(\{\eta_a,\eta_b\}\eta_c\bigr).
		\end{equation}
		The matrices $F^\eta_a,D^\eta_a$ and the maps
		$\Theta^{\eta,\pm}$ are defined from $f^\eta,d^\eta$,  exactly as in
		\eqref{eq:F_and_D}--\eqref{eq:theta-definitions}.
	\end{definition}
	
	Clearly, the generalized Gell--Mann basis
	$(\lambda_1,\dots,\lambda_s)$ is a special case, with
	$f^\lambda=f$ and $d^\lambda=d$.
	
	\begin{lemma}\label{lem:frame_relations}
		For  $\eta$  in Definition \ref{def:frame},  $f^\eta$ is real and completely antisymmetric,
		$d^\eta$ is real and completely symmetric, and
		\begin{equation}\label{eq:frame_CCR}
			\begin{aligned}
				[\eta_a,\eta_b]
				&=
				2\ii\sum_{c=1}^s f^\eta_{abc}\eta_c,
				\\
				\{\eta_a,\eta_b\}
				&=
				\frac4n\delta_{ab}I_n
				+
				2\sum_{c=1}^s d^\eta_{abc}\eta_c.
			\end{aligned}
		\end{equation}
		Moreover, if
		\beq \label{eq:basis_transform}
		\eta_b=\sum_{a=1}^sR_{ab}\lambda_a
		\eeq
	for some orthogonal matrix $R\in \mathbb{R}^{s\times s}$, then the transformed and original structured tensors are related by
		\begin{equation}\label{eq:tensor_transform}
			\begin{aligned}
				f^\eta_{abc}
				&=
				\sum_{a',b',c'=1}^s
				f_{a'b'c'}R_{a'a}R_{b'b}R_{c'c},\\
				d^\eta_{abc}
				&=
				\sum_{a',b',c'=1}^s
				d_{a'b'c'}R_{a'a}R_{b'b}R_{c'c}.
			\end{aligned}
		\end{equation}
	\end{lemma}
	
	\begin{IEEEproof}
		Since each $\eta_a$ is Hermitian,
		$\Tr([\eta_a,\eta_b]\eta_c)$ is purely imaginary, whereas
		$\Tr(\{\eta_a,\eta_b\}\eta_c)$ is real. Thus, both   $f^\eta$ and $d^\eta$ are real.  The antisymmetry of $f^\eta$ and 
		symmetry of $d^\eta$  are immediate consequences of \eqref{eq:frame_tensors_def}. The operator $\frac1{2\ii}[\eta_a,\eta_b]$ is traceless Hermitian, and by \eqref{eq:frame_tensors_def} its
		coefficient along $\eta_c$ is
		\[
		\left\langle
		\eta_c,\frac1{2\ii}[\eta_a,\eta_b]
		\right\rangle_{\rm nHS}
		=
		f^\eta_{abc}.
		\]
		Similarly,  it can be easily shown that $\{\eta_a,\eta_b\}$ has trace $4\delta_{ab}$, the identity component
		$\frac4n\delta_{ab}I_n$, and traceless component \(\eta^c\) with coefficient
		$2d^\eta_{abc}$.  This proves \eqref{eq:frame_CCR}.  Eq. \eqref{eq:tensor_transform} follows by substituting the basis expansion \eqref{eq:basis_transform} into
		\eqref{eq:frame_tensors_def}.
	\end{IEEEproof}
	
	\begin{definition}\label{def:rel_admissible}
		Given $\eta$  in Definition \ref{def:frame},  a matrix $T\in\mathbb R^{s\times s}$ is called
		\emph{$(\lambda\!\to\!\eta)$-admissible} if $T^\top T=I_s$, and for every
		$\beta\in\mathbb C^s$,
		\begin{equation}\label{eq:rel_admissible}
			\begin{aligned}
				\Theta^{\eta,-}(T^\top\beta)
				&=
				T^\top\Theta_n^-(\beta)T,
				\\
				\Theta^{\eta,+}(T^\top\beta)
				&=
				T^\top\Theta_n^+(\beta)T.
			\end{aligned}
		\end{equation}
		Definition~\ref{def:admissible} is recovered when $\eta=\lambda$.
	\end{definition}
	
	The following result shows that a $(\lambda\!\to\!\eta)$-admissible map transforms the generalized Gell-Mann matrices to observables which preserve commutation and anticommutation relations, relative to the basis $\eta$.
	
	\begin{lemma}\label{lem:rotation_frame}
		Given $\eta$  in Definition \ref{def:frame} and let $T\in\mathbb R^{s\times s}$.  Then the transformed observables
		$\widetilde X=T^\top X$ satisfy
		\begin{equation}\label{eq:frame_XX}
			\begin{aligned}
				[\widetilde X,\widetilde X^\top]
				&=
				2\ii\Theta^{\eta,-}(\widetilde X),
				\\
				\{\widetilde X,\widetilde X^\top\}
				&=
				\frac4nI_s+2\Theta^{\eta,+}(\widetilde X)
			\end{aligned}
		\end{equation}
		if and only if $T$ is $(\lambda\!\to\!\eta)$-admissible.
	\end{lemma}
	
	\begin{IEEEproof}
		The proof is identical to that of Lemma~\ref{lem:rotation}, with 
		matrices $F_k,D_k$ replaced by $F^\eta_k,D^\eta_k$.  It uses only the
		linear independence of the matrices $I_n,\lambda_1,\dots,\lambda_s$ and the 
		antisymmetry and symmetry established in Lemma~\ref{lem:frame_relations}.
	\end{IEEEproof}
	
	The following lemma gives the construction of   $(\lambda\!\to\!\eta)$-admissible maps, which is a counterpart of Lemma \ref{lem:T}.
	
	\begin{lemma}\label{lem:T_frame}
		Given $\eta$  in Definition \ref{def:frame},   and et $U\in SU(n)$, and define
		$T\in\mathbb R^{s\times s}$ by
		\begin{equation}\label{eq:T_frame}
			U^\dagger\eta_jU
			=
			\sum_{k=1}^sT_{kj}\lambda_k,
		\end{equation}
			where
		$T_{kj}
		=
		\langle\lambda_k,U^\dagger\eta_jU\rangle_{\rm nHS}$.  Then $T$ is $(\lambda\!\to\!\eta)$-admissible.
	\end{lemma}
	
	\begin{IEEEproof}
		Clearly, the family
		$(U^\dagger\eta_1U,\dots,U^\dagger\eta_sU)$ is an orthonormal basis of
		$\Hermzero(n)$, so $T^\top T=I_s$.  Furthermore,
		\beq \label{eq:temp1_21aug}
		\widetilde X_j 
		=
		e_j^\top T^\top X
		=
		(Te_j)^\top X
		=
		\Phi(Te_j)
		=
		U^\dagger\eta_jU.
		\eeq
		Unitary conjugation of \eqref{eq:frame_CCR} therefore gives
		\begin{align*}
			[\widetilde X_a,\widetilde X_b]
			&=
			2\ii\sum_c f^\eta_{abc}\widetilde X_c,
			\\
			\{\widetilde X_a,\widetilde X_b\}
			&=
			\frac4n\delta_{ab}I_n
			+
			2\sum_c d^\eta_{abc}\widetilde X_c,
		\end{align*}
	which establish \eqref{eq:frame_XX} via  \eqref{eq:frame_CCR}.  The result then follows from
		Lemma~\ref{lem:rotation_frame}.
	\end{IEEEproof}
	
	Lemma \ref{lem:T_frame} tells us that given the generalized Gell-Mann matrices  $\{\lambda_1,\ldots, \lambda_s\}$ and another orthonormal basis $\{\eta_1,\ldots, \eta_s\}$, each $U\in SU(n)$ induces a $(\lambda\!\to\!\eta)$-admissible map $T$.
	
	The following result, echoing Lemma \ref{lem:rotation_frame},  says that a $(\lambda\!\to\!\eta)$-admissible map  indeed transforms the quantum $n$-level system \eqref{eq:irka_sys2}  to a  quantum $n$-level system relative to the  orthonormal basis $\{\eta_1,\ldots, \eta_s\}$.
	
	\begin{lemma}\label{lem:PR_frame}
		Suppose the matrices of the system \eqref{eq:irka_sys2} are given by
		\eqref{eq:N-from-C} for physical parameters
		$\theta=(\alpha,C_1,C_2)$ in \eqref{eq:physical_parameters}, and let a matrix   $T\in\mathbb{R}^{s\times s}$ be
		$(\lambda\!\to\!\eta)$-admissible.  Define $
		\widetilde\alpha=T^\top\alpha$,
		$\widetilde C_i=C_iT$, and $
		\widetilde\Lambda=\Lambda T$.
		Then $
		\widetilde A_0=T^\top A_0$,
		$\widetilde A=T^\top AT$, and
		$\widetilde B_\mu=T^\top B_\mu T$ 
		are the counterpart of the matrices in \eqref{eq:N-from-C} under the orthonormal basis $\eta$.   
		Moreover, the Hamiltonian and coupling operators are
		\begin{equation}\label{eq:HL_frame_invariance}
			H
			=
			\alpha^\top X
			=
			\widetilde\alpha^\top\widetilde X,
			\qquad
			L
			=
			\Lambda X
			=
			\widetilde\Lambda\widetilde X.
		\end{equation}
	\end{lemma}
	
	\begin{IEEEproof}
		Using \eqref{eq:rel_admissible}, $
		\widetilde B_{1k}
		=
		\Theta^{\eta,-}\bigl((\widetilde C_2)_k\bigr)$, and
		$\widetilde B_{2k}
		=
		-\Theta^{\eta,-}\bigl((\widetilde C_1)_k\bigr)$. 
		Since $T$ is real,
		$\widetilde\Lambda^\#=\Lambda^\#T$.  Using
		$TT^\top=I_s$ and \eqref{eq:rel_admissible} for complex arguments gives
		the $\eta$-basis forms of $R_k$ and $Q_k$, and hence the $\eta$-basis
		version of \eqref{eq:A-full}.   Finally,  as $	\widetilde\alpha^\top\widetilde X
			=
			\alpha^\top TT^\top X
			=
			\alpha^\top X$, and $\widetilde\Lambda\widetilde X
			=
			\Lambda TT^\top X
			=
			\Lambda X$,
		which prove \eqref{eq:HL_frame_invariance}.
	\end{IEEEproof}
	
	\begin{definition}\label{def:block_frame}
		Let $\mathcal S\subseteq\mathbb R^s$ be strongly aligned in the sense of
		Definition~\ref{def:strong_alignment},  and with Artin--Wedderburn data
		$\{(r_k,m_k)\}_{k=1}^K$ as given in \eqref{eq:AW data}. Set $
		p\triangleq\sum_{k=1}^K(r_k^2-1)$. 
		An orthonormal basis $\eta$ of $\Hermzero(n)$  is called \emph{block adapted} if
		\beq\label{eq:span_D}
		\Span_{\mathbb R}\{\eta_1,\dots,\eta_p\}
		=
		\mathfrak D,
		\eeq
		and
		\beq\label{eq:span_Z_0}
		\Span_{\mathbb R}
		\{\eta_{p+1},\dots,\eta_{p+K-1}\}
		=
		\mathfrak Z_0,
		\eeq
		where $\mathfrak D$ and $\mathfrak Z_0$  are given in \eqref{eq:july25_4.4} and \eqref{eq:july25_4.3}, respectively.
	\end{definition}
	
	\bmrk
	Block-adapted bases always exist.  By Lemma~\ref{lem:D_intersection},
	$\mathfrak D$ and $\mathfrak Z_0$ are mutually orthogonal subspaces of
	$\Hermzero(n)$.  One may therefore choose orthonormal bases  of these two
	subspaces and complete their union to get an orthonormal basis of
	$\Hermzero(n)$.
	\emrk
	
	The following is the main result of this section. 
	
	\begin{theorem}
		\label{thm:decomp_frame}
		Assume the hypotheses of Theorem~\ref{thm:decomp}.  Let $U\in SU(n)$ realize
			\eqref{eq:july25_4.1}. Then let $\eta$ be the associated block-adapted basis, and define the corresponding $T$ as in Lemma \ref{lem:T_frame}. The following statements hold.
		
		\begin{enumerate}[label=(\roman*)]
			\item
			$T$ is $(\lambda\!\to\!\eta)$-admissible.  Thus, the transformed observables 
			$\widetilde X=T^\top X$ satisfies the  commutation and
			anticommutation relations \eqref{eq:frame_XX}, and the transformed system is physically
			realizable with respect to the block-adapted basis $\eta$.
			
			\item
			The first \(p\) columns of \(T\) form an orthonormal basis of
			\(\mathcal S\).  Columns \(p+1,\ldots,p+K-1\) form an orthonormal
			basis of \(\mathcal Q\).  The remaining
			\(s-p-K+1\) columns form an orthonormal basis of $
			(\mathcal S\oplus\mathcal Q)^\perp$. 
			
			\item
			Partition
			\beq \label{eq:temp_T_21aug}
			T=
			\begin{bmatrix}
				T_{\mathcal S}&T_\perp
			\end{bmatrix},
			\qquad
			T_{\mathcal S}\in\mathbb R^{s\times p}.
			\eeq
			The transformed system matrices have the block forms
			\begin{subequations}\label{eq:diagonal}
				\begin{align}
					\widetilde A
					&=
					T^\top AT
					=
					\begin{bmatrix}
						A_{\mathcal S}&0\\
						0&A_\perp
					\end{bmatrix},
					\label{eq:july25_5.8}\\
					\widetilde B_\mu
					&=
					T^\top B_\mu T
					=
					\begin{bmatrix}
						0&0\\
						0&B_{\mu,\perp}
					\end{bmatrix},
					\label{eq:july25_5.9}\\
					\widetilde C
					&=
					CT
					=
					\begin{bmatrix}
						0&C_\perp
					\end{bmatrix},
					\label{eq:july25_5.10}\\
					\widetilde A_0
					&=
					T^\top A_0
					=
					\begin{bmatrix}
						0\\
						A_{0,\perp}
					\end{bmatrix}.
					\label{eq:july25_5.11}
				\end{align}
			\end{subequations}
			
			\item
			Denote the transformed system observables by 
		$
			\widetilde X 	=\left[
\begin{smallmatrix}
				\widetilde X_{\mathcal S}\\
				\widetilde X_\perp
\end{smallmatrix}\right]= T^\top X$.
			Then
			\begin{subequations}\label{eq:block_sector_dynamics}
				\begin{align}
					d\widetilde X_{\mathcal S}
					&=
					A_{\mathcal S}\widetilde X_{\mathcal S}\,dt,
					\label{eq:july26_8.4}\\
					d\widetilde X_\perp
					&=
					\bigl(A_{0,\perp}
					+A_\perp\widetilde X_\perp\bigr)\,dt
					+
					\sum_{\mu=1}^m
					B_{\mu,\perp}\widetilde X_\perp\,dW_\mu,
					\label{eq:july26_8.5}\\
					dY
					&=
					C_\perp\widetilde X_\perp\,dt+dW.
					\label{eq:july26_8.6}
				\end{align}
			\end{subequations}
			Clearly, the subsystem in \eqref{eq:july26_8.4} is a DFS of the finite-level quantum system \ref{eq:irka_sys2}  in the $\eta$-basis.
		\end{enumerate}
	\end{theorem}
	
	\begin{IEEEproof}
		Part (i) follows from Lemmas~\ref{lem:T_frame} and \ref{lem:PR_frame}.
		
		For Part (ii), if $j\le p$, then by \eqref{eq:span_D} we get $\eta_j\in\mathfrak D$.  Moreover, by \eqref{eq:july25_4.7},  $
		U^\dagger\mathfrak DU=\Phi(\mathcal S)$, 
		which, together with \eqref{eq:temp1_21aug}, gives
		\beq \label{eq:31aug_temp1}
		\Phi(Te_j)
		=
		U^\dagger\eta_jU
		\in
		\Phi(\mathcal S).
		\eeq
		Thus $Te_j\in\mathcal S$.   Clearly, these $p$ columns are orthonormal, and as  $
		p=\dim_{\mathbb R}\mathfrak D=\dim\mathcal S$,   they form an orthonormal basis of $\mathcal S$.  The assertion concerning
		$\mathcal Q$ follows in the same way from
		$U^\dagger\mathfrak Z_0U=\Phi(\mathcal Q)$ in Lemma \ref{lem: Z_0 and D}.
		
		For Part (iii), by \eqref{eq:temp_T_21aug}  and item (ii) proved above,
		\beq \label{eq:P and perp_21aug}
		P_{\mathcal S}
		=
		T_{\mathcal S}T_{\mathcal S}^\top,
		\qquad
		P_{\mathcal S}T_{\mathcal S}=T_{\mathcal S},
		\qquad
		T_\perp^\top P_{\mathcal S}=0.
		\eeq
		Thus, by  \eqref{eq:july25_5.4},  $
		T_\perp^\top AT_{\mathcal S}
		=
		T_\perp^\top(I_s-P_{\mathcal S})
		AP_{\mathcal S}T_{\mathcal S}
		=
		0$,	
		and
		$
		T_{\mathcal S}^\top AT_\perp
		=
		T_{\mathcal S}^\top P_{\mathcal S}
		A(I_s-P_{\mathcal S})T_\perp
		=
		0$. 
		Therefore $\widetilde A$ has the block form \eqref{eq:july25_5.8}. Next, by \eqref{eq:july25_5.5} and \eqref{eq:P and perp_21aug}, $
		B_\mu T_{\mathcal S}
		=
		B_\mu P_{\mathcal S}T_{\mathcal S}
		=
		0$. 
		Moreover, $B_\mu^\top=-B_\mu$ implies
		$P_{\mathcal S}B_\mu=0$, and hence $
		T_{\mathcal S}^\top B_\mu
		=
		T_{\mathcal S}^\top P_{\mathcal S}B_\mu
		=
		0$. This proves the form $\widetilde B_\mu$ in \eqref{eq:july25_5.9}.  Similarly,  by \eqref{eq:july25_5.6}--\eqref{eq:july25_5.7}, and \eqref{eq:P and perp_21aug}, $
		CT_{\mathcal S}
		=
		CP_{\mathcal S}T_{\mathcal S}
		=
		0$,
		and
		$
		T_{\mathcal S}^\top A_0
		=
		T_{\mathcal S}^\top P_{\mathcal S}A_0
		=
		0$. Thus \eqref{eq:july25_5.10}--\eqref{eq:july25_5.11}  hold.  Part (iv) follows by
		substitution of the above transformed matrices into the transformed system.
	\end{IEEEproof}

\subsection{DFS under the Gell--Mann basis}
	\label{subsec:fixed_frame_corollary}
	
The following result, which is a consequence of Theorem \ref{thm:decomp_frame}, constructs a DFS of the finite-level quantum system \ref{eq:irka_sys2}   with respect to the generalized Gell--Mann basis.
	
	\begin{corollary} \label{cor:fixed_frame_decomp}
		Assume the hypotheses of Theorem~\ref{thm:decomp}.  Let
		\(U\in SU(n)\) realize the Artin--Wedderburn decomposition
		\eqref{eq:july25_4.1}, define $
		T=\Ad_{U^\dagger}$, 
		and set $
		p=\dim_{\mathbb R}\mathfrak D
		=
		\sum_{k=1}^K(r_k^2-1)$. 
		Suppose that \(\mathfrak D\) is spanned by a subset of the selected
		generalized Gell--Mann basis.  After relabeling that subset, we assume
		\begin{equation}
			\mathfrak D
			=
			\Span_{\mathbb R}\{\lambda_1,\ldots,\lambda_p\}.
			\label{eq:fixed_frame_compatibility}
		\end{equation}
		Then \(T\) is admissible in the sense of
		Definition~\ref{def:admissible}, its first \(p\) columns form an
		orthonormal basis of \(\mathcal S\), and the transformed matrices have
		the block forms \eqref{eq:diagonal}.  Consequently, the transformed
		system has the   form \eqref{eq:block_sector_dynamics} while
		retaining the original generalized Gell--Mann structure tensors
		\((f,d)\).
		
		A sufficient condition for
		\eqref{eq:fixed_frame_compatibility} is the following: \emph{exactly one}
		Artin--Wedderburn block has \(r_k\ge2\)  with multiplicity
		\(m_k=1\), and the blocks are ordered so that that block occupies the leading
		diagonal corner.
	\end{corollary}
	
	\begin{IEEEproof}
		By Lemma~\ref{lem:T}, \(T=\Ad_{U^\dagger}\) is admissible.  For
		\(j=1,\ldots,p\),  
		\eqref{eq:fixed_frame_compatibility} gives
		\(\lambda_j\in\mathfrak D\).  Hence, by
		Lemma~\ref{lem: Z_0 and D},
		\beq \label{eq:Phi_T_e_j}
		\Phi(Te_j)
		=
		U^\dagger\lambda_jU
		\in
		U^\dagger\mathfrak DU
		=
		\Phi(\mathcal S),
		\quad j=1,\ldots,p.
		\eeq
		Therefore \(Te_j\in\mathcal S\).  Since \(T\) is orthogonal, the
		vectors \(Te_1,\ldots,Te_p\) are orthonormal.  Moreover, by \eqref{eq:july25_4.7},
		$
		p=\dim_{\mathbb R}\mathfrak D=\dim\mathcal S,
		$ 
		so they form an orthonormal basis of \(\mathcal S\).  The block forms then follow from the
		same projector calculation used in the proof of items (iii) and (iv) of 
		Theorem~\ref{thm:decomp_frame}.
		
		For the sufficient condition, let that \emph{unique}  block have
		size \(r\ge2\) and multiplicity one, and be placed in the leading
		\(r\times r\) corner.  Every other block has \(r_k=1\), and hence
		\(\Hermzero(r_k)=\{0\}\).  It follows that $
		\mathfrak D
		=
		\Hermzero(r)\oplus0_{n-r}$. 
		The generalized Gell--Mann matrices
		$
		\{u_{jk},v_{jk}:1\le j<k\le r\}
		\cup
		\{w_1,\ldots,w_{r-1}\}
		$
		are supported on the leading \(r\times r\) corner and form an
		orthonormal basis of this subspace.  Indeed, their number is
		$
		2\binom r2+(r-1)=r^2-1
		=
		\dim_{\mathbb R}\Hermzero(r).
		$ 
		After relabeling these matrices as
		\(\lambda_1,\ldots,\lambda_p\),  
		\eqref{eq:fixed_frame_compatibility} holds.
	\end{IEEEproof}
	
%
	
	\section{QND observables}  \label{sec:QND}
	
	In this section, we study QND variables of the quantum finite-level system  \eqref{eq:irka_sys2}. 
	
	\subsection{QND observables in the sense of linear systems theory}
	
	Given a matrix
$
		V=
		\begin{bmatrix}
			v_1 & \cdots & v_r
		\end{bmatrix}
		\in\R^{s\times r}
	$ (\(r\leq s\)),  define a  vector of observables
	\begin{equation}
		Z
		\triangleq
		V^{\top}X
		=
		\begin{bmatrix}
			Z_1\\
			\vdots\\
			Z_r
		\end{bmatrix}.
		\label{eq:july27_qnd_1.7}
	\end{equation}
	Left-multiplying system \eqref{eq:irka_sys2} by \(V^{\top}\) gives
	the following QSDE
	\begin{equation}
		dZ
		=
		\left(
		V^{\top}A_0
		+
		V^{\top}AX
		\right)dt
		+
		\sum_{\mu=1}^{m}
		V^{\top}B_\mu X\,dW_\mu.
		\label{eq:july27_qnd_4.2}
	\end{equation}
	
	\begin{theorem}
		\label{thm:july27_qnd_G7}
		Let \(V\in\R^{s\times r}\) have full column rank, and define $
		Z=V^{\top}X$ as in \eqref{eq:july27_qnd_1.7}.  Then the following statements hold.
		
		\begin{enumerate}[label=(\roman*)]
			\item The components of the vector of observables  \(Z\) mutually commute if and only if 
			\begin{equation}
				\Theta_n^-(v_i)v_j=0,		\qquad
				i,j=1,\ldots,r.
				\label{eq:july27_qnd_2.6}
			\end{equation}
			
			\item In addition,  if
			\begin{equation}
				\Theta_n^+(v_i)v_j
				\in
				\Range(V),
				\qquad
				i,j=1,\ldots,r
				\label{eq:july27_qnd_3.4}
			\end{equation}
			holds, then
			\begin{equation} \label{eq:july27_qnd_3.3}
				\mathfrak A_Z \triangleq \Span_{\C}\{I_n,Z_1,\ldots,Z_r\}
			\end{equation}
			is a commutative unital \(C^*\)-algebra.
			
			\item If 
			\begin{equation}
				V^{\top}B_\mu=0,
				\qquad
				\mu=1,\ldots,m,
				\label{eq:july27_qnd_4.3}
			\end{equation} 
			and 
			\begin{equation}
				A^\top\Range(V)
				\subseteq
				\Range(V)
				\label{eq:july27_qnd_6.2}
			\end{equation} 
			hold, then the QSDE \eqref{eq:july27_qnd_4.2} reduces to
			\begin{equation}
				dZ=(b+KZ)\,dt,
				\label{eq:july27_qnd_9.5}
			\end{equation}
			where the matrices
			\begin{equation}
				K
				=
				V^{\top}AV
				(V^{\top}V)^{-1}, \quad b=V^{\top}A_0 .
				\label{eq:july27_qnd_6.9}
			\end{equation}	
			In addition, if \eqref{eq:july27_qnd_2.6} holds, then
			\begin{equation}
				[Z_i(t),Z_j(s)]=0
				\label{eq:july27_qnd_9.6}
			\end{equation}
			for every \(i,j\) and all \(t,s\geq 0\). 
			
			\item Moreover,  denote $Y_M = MY$ for some \(M\in\mathbb R^{p\times m}\) (\(p\leq m\)). If
			\begin{equation} \label{eq::MJM}
				MJM^\top=0,
			\end{equation}
			and
			\begin{equation} \label{eq:KerV}
				\ker(V^\top)\subseteq\ker(MC)
			\end{equation}
			hold, then
			\begin{equation} \label{eq:Y_M}
				dY_M=C_ZZ\,dt+M\,d\mathcal W
			\end{equation}
			describes the dynamics of the commuting observables $Y_M$, where 			
			\begin{equation} \label{eq:C_Z}
				C_Z=MCV(V^\top V)^{-1} \in \mathbb{R}^{p\times r},
			\end{equation}	
			
			\item Finally, define the observability matrix for the linear dynamics described by \eqref{eq:july27_qnd_9.5} and \eqref{eq:Y_M},
			\begin{equation} \label{eq:O_Z}
				\mathcal O_Z
				\triangleq
				\begin{bmatrix}
					C_Z\\
					C_ZK\\
					\vdots\\
					C_ZK^{r-1}
				\end{bmatrix}.
			\end{equation}	
			If 
			\begin{equation} \label{eq:rank_Oz}
				\operatorname{rank}\mathcal O_Z=r,
			\end{equation}	
			then the vector \(Z\) is observable from \(Y_M\) in the
			standard linear-systems sense.
		\end{enumerate}
	\end{theorem}
	
	\begin{IEEEproof}
		(i) From \eqref{eq:CCR}  we have
$
			[Z,Z^\top]
			=
			[V^{\top}X,X^{\top}V]
			=
			V^\top[X,X^\top]V =2\ii 
			V^\top\Theta_n^-(X)V$. 
		Thus, similar to the proof for Lemma \ref{lem:rotation},  
		\begin{equation}
			[Z_i,Z_j]=0,
			\qquad
			i,j=1,\ldots,r.
			\label{eq:july27_qnd_2.8}
		\end{equation}
		if and only if \eqref{eq:july27_qnd_2.6} holds.
		
		(ii) 	Consider the space   \(\mathfrak A_Z\)  defined in \eqref{eq:july27_qnd_3.3} . Under item (i) proved  above and the matrix operator norm $
		\|A\| = 
		\sigma_{\rm max}(A)$, to show that  \(\mathfrak A_Z\)   is commutative unital \(C^*\)-algebra, it is necessary and
		sufficient  to show that   \(\mathfrak A_Z\) is  closed under the anticommutation relations. 	Since \(V\) has full column rank, by
		\eqref{eq:july27_qnd_3.4} there is a unique vector
		\(\Gamma_{ij}\in\R^r\) satisfying
	$
			\Theta_n^+(v_i)v_j
			=
			V\Gamma_{ij}$, 
		where the matrices
$
			\Gamma_{ij}
			=
			(V^{\top}V)^{-1}
			V^\top
			\Theta_n^+(v_i)v_j.
$ 
		Consequently, by \eqref{eq:CCR_ACR} and \eqref{eq:july27_qnd_1.7},
		\begin{equation}
			\{Z_i,Z_j\}
			=
			\frac{4}{n}
			v_i^{\top}v_j\,I_n
			+
			2\Gamma_{ij}^\top Z.
			\label{eq:july27_qnd_3.7}
		\end{equation}
		As a result, if both \eqref{eq:july27_qnd_2.6}--\eqref{eq:july27_qnd_3.4} hold, then by \eqref{eq:july27_qnd_2.8}--\eqref{eq:july27_qnd_3.7},
	$
			Z_iZ_j
			=
			\frac{2}{n}
			v_i^{\top}v_j\,I_n
			+
			\Gamma_{ij}^\top Z \in \mathfrak A_Z.
	$
		
		(iii) As $V$ is of full column rank, by \eqref{eq:july27_qnd_6.2},  the matrix  \(K\in\R^{r\times r}\) in \eqref{eq:july27_qnd_6.9} satisfies  $
		V^{\top}A=KV^\top$.
		Then by \eqref{eq:july27_qnd_1.7}
		\beq \label{eq:KZ}
		V^{\top}AX
		=
		KV^{\top}X
		=
		KZ.
		\eeq
		By \eqref{eq:july27_qnd_4.3} and \eqref{eq:KZ}, the  QSDE \eqref{eq:july27_qnd_4.2} reduces to  \eqref{eq:july27_qnd_9.5}.
		The solution of \eqref{eq:july27_qnd_9.5} with the initial condition $Z(0) = Z$ is
		\begin{equation*}
			Z(t)
			=
			e^{Kt }Z
			+
			\int_{0}^{t}
			e^{K(t-\tau)}b\,d\tau.
			\label{eq:july27_qnd_4.7}
		\end{equation*}
		Consequently,
		\begin{align}
			[Z(t),Z(s)^\top]
			&=
			e^{Kt}
			[Z,Z^\top ]
			e^{K^{\top}s}.
			\label{eq:july27_qnd_5.1}
		\end{align}
		As proved in item (i) above, if \eqref{eq:july27_qnd_2.6} holds, then we have $
		[Z,Z^\top]=0$, whose substitution into \eqref{eq:july27_qnd_5.1} gives
	$
			[Z(t),Z(s)^\top]=0$ for all
			$t,s\geq 0.
$ 
		Componentwise, \eqref{eq:july27_qnd_9.6} holds for all
			$i,j=1,\ldots,r$ and 
		$
			t,s\geq 0$. 
			
		(iv) Under \eqref{eq:KerV}, there exists a  matrix $C_Z$ such that  
	$
		MC = C_Z V^\top.
	$ 
		In fact, $C_Z$ defined in \eqref{eq:C_Z} is a valid choice.  As a result, we get \eqref{eq:Y_M}.  Moreover, by  \eqref{eq::MJM},  the components of  the output observable vector $ Y_M$ mutually commute.
		
		(v) The conclusion follows as in the well-known linear systems theory \cite{ZDG96}. 
	\end{IEEEproof}			
	
	\subsection{QND observables in the sense of spectral separation}
	
	According to Theorem \ref{thm:july27_qnd_G7}, if \eqref{eq:rank_Oz} does not hold, the observables $\{Z_1,\ldots, Z_r\}$  in \eqref{eq:july27_qnd_9.5} cannot be  completely reconstructed from the output field vector  \eqref{eq:Y_M}, in the sense of linear systems theory  \cite{ZDG96}.   However, for the  finite-dimensional 
	\(C^*\)-algebra \(\mathfrak A_Z\) of commuting observables defined in \eqref{eq:july27_qnd_3.3},  a different notion of observability is often more
	appropriate. Specifically, instead of asking whether  each \(Z_i\) can
	be reconstructed linearly from \(Y_M\), one may ask whether a vector of selected output-signal 
	observables (to bee defined in \eqref{eq:F_ell} below) is able to distinguish  joint spectral subspaces of
	 \(\mathfrak A_Z\). Equivalently, one asks whether  a vector of selected output-signal 
	observables is able to generate  algebra
	\(\mathfrak A_Z\). In this sense, we say the  commutative observable algebra
	\(\mathfrak A_Z\) is observable. This operator-algebraic notion of observability is developed in this subsection.
	
	According to Theorem~\ref{thm:july27_qnd_G7}, if
	\eqref{eq:july27_qnd_2.6}--\eqref{eq:july27_qnd_3.4} hold, then all the elements in 
	\(\mathfrak A_Z\) defined in \eqref{eq:july27_qnd_3.3}  are mutually commutative. Consequently,
	there is a  family of  minimal
	mutually orthogonal projectors
 	\(\Pi_1,\ldots,\Pi_q\in\mathfrak A_Z\) such that
	\begin{equation}
		\begin{gathered}
			\mathfrak A_Z
			=
			\operatorname{span}_{\mathbb C}
			\{\Pi_1,\ldots,\Pi_q\},
			\\
			\Pi_\alpha\Pi_\beta
			=
			\delta_{\alpha\beta}\Pi_\alpha,
			\qquad
			\sum_{\alpha=1}^q\Pi_\alpha=I_n.
		\end{gathered}
		\label{eq:A_Z_Pi}
	\end{equation}

Recall the selected output equation \eqref{eq:Y_M}, and define
	the \emph{output-signal} observables
	\beq \label{eq:F_ell}
	F_\ell
	=
	\sum_{j=1}^r(C_Z)_{\ell j}Z_j
	\in\mathfrak A_Z,
	\qquad \ell=1,\ldots,p.
	\eeq
	These are system observables, rather than output-field
	quadratures. The terminology ``output-signal observables'' refers to their role in the
	selected output equation. Actually from the  output equation \eqref{eq:Y_M} we have
	\[
	dY_{M,\ell}(t)
	=
	F_\ell(t)\,dt+(M\,d\mathcal W(t))_\ell, \quad \ell=1,\ldots,p.
	\]
	Thus, $F_\ell(t)$ is the system-dependent signal 
	in the measured-output QSDE 
	\cite[Sec.~2.1]{GZP19}.
	The spectral-separation criterion below concerns the system
	algebra generated by these signal observables; it does not
	assert exact reconstruction from a finite noisy measurement
	record.
	
	\bthm \label{thm:QND_2}
	The selected  mutually commuting output-signal observables \(\{F_1,\ldots, F_p\}\) generate $\mathfrak A_Z$ if and only if the joint signal-value vectors (namely $f_{j\alpha}$ for each observable $F_j$, \(j=1,\ldots, p\))
	\beqm
	\mathbf f_\alpha
	=
	\begin{bmatrix}
		f_{1\alpha}\\
		\vdots\\
		f_{p\alpha}
	\end{bmatrix}, \quad \alpha =1,\ldots q
	\eeqm
	satisfy
	\beq \label{eq:nonparallel}
	\mathbf f_\alpha\neq\mathbf f_\beta
	\quad
	\mathrm{whenever } \quad \alpha\neq\beta.
	\eeq
Equivalently, the $*$-algebra 
		\beq \label{eq:alg_A_Z}
		\operatorname{alg}^{*}(I_n,F_1,\ldots,F_p)
		=
		\mathfrak A_Z.
		\eeq
	\ethm
	
	\begin{IEEEproof} 
		(\emph{Sufficiency}.) The proof is algebraic. By means of 
		\eqref{eq:A_Z_Pi} and the form of \(F_\ell\), direct algebraic manipulations yield
		\begin{align*}
			\sum_{\ell=1}^p
			\left(F_\ell-f_{\ell\beta}I_n\right)^2
			&=
			\sum_{\gamma=1}^q
			\Bigl[
			\sum_{\ell=1}^p
			(f_{\ell\gamma}-f_{\ell\beta})^2
			\Bigr]\Pi_\gamma
			\\
			&= 
			\sum_{\gamma=1}^q
			\|\mathbf f_\gamma-\mathbf f_\beta\|^2\Pi_\gamma,
		\end{align*}
		from which we have
		\[
		\frac{
			\sum_{\ell=1}^p(F_\ell-f_{\ell\beta}I_n)^2
		}{
			\|\mathbf f_\alpha-\mathbf f_\beta\|^2
		}
		=
		\sum_{\gamma=1}^q
		\frac{
			\|\mathbf f_\gamma-\mathbf f_\beta\|^2
		}{
			\|\mathbf f_\alpha-\mathbf f_\beta\|^2
		}
		\Pi_\gamma.
		\]
		This, together with
	\[
		\Bigl(\sum_\gamma a_\gamma\Pi_\gamma\Bigr)
		\Bigl(\sum_\delta b_\delta\Pi_\delta\Bigr)
		=
		\sum_\gamma a_\gamma b_\gamma\Pi_\gamma,
		\]
		yields
		\beq \label{eq:Pi_PI_alpha}
		\prod_{\substack{\beta=1\\\beta\neq\alpha}}^q
		\frac{
			\displaystyle\sum_{\ell=1}^p
			\left(F_\ell-f_{\ell\beta}I_n\right)^2
		}{
			\|\mathbf f_\alpha-\mathbf f_\beta\|^2 
		}
		=	\Pi_\alpha, \quad \alpha = 1\ldots, q,
		\eeq
		where  \eqref{eq:nonparallel} has been used.  By \eqref{eq:Pi_PI_alpha}, all orthogonal projectors $\Pi_\alpha$ can be obtained from the output-signal observables \(\{F_1,\ldots, F_p\}\). 
		
		(\emph{Necessity}.) We prove it by contradiction.  	Suppose that
		\(\mathbf f_\alpha=\mathbf f_\beta\) for some
		\(\alpha\neq\beta\).    As the observables $F_\ell$ mutually commute, 
	every element of $\mathfrak B\triangleq \operatorname{alg}^{*}(I_n,F_1,\ldots,F_p)$ is a polynomial in them, and by mean of spectral decomposition such polynomial has the form
	\[
	P(F_1,\ldots,F_p)
	=
	\sum_{\gamma=1}^{q}P(\mathbf f_\gamma)\Pi_\gamma.
	\]
	If $\Pi_\alpha=P(F_1,\ldots,F_p)$ for some polynomial \(P(F_1,\ldots,F_p)\), multiplication by
	$\Pi_\alpha$ and $\Pi_\beta$, respectively, yield $
	P(\mathbf f_\alpha)=1$, and 
	$P(\mathbf f_\beta)=0$, contradicting $\mathbf f_\alpha=\mathbf f_\beta$.
	The same argument excludes $\Pi_\beta$ too.  Thus, neither $\Pi_\alpha$ nor $\Pi_\beta$ belongs to
	$\mathfrak B$. On the other hand,
 both  $\Pi_\alpha$ and $\Pi_\beta$ belong to $\mathfrak A_Z$.  We conclude
	that $\mathfrak B$ is a proper subalgebra of $\mathfrak A_Z$, therefore reaching an contradiction.
\end{IEEEproof}
	
	\section{BAE measurement} \label{sec:BAE}

In this section, we study quantum  BAE measurement for  the \(n\)-level quantum system 
\eqref{eq:irka_sys2}.  A sufficient condition is given by the following result.
	
	\begin{theorem} \label{thm:july27_G8_general}
		Consider system \eqref{eq:irka_sys2}. 
		Define input field observables,  output field observables and output matrix
		\beq \label{eq:WYC}
		\mathcal W_g=g^\top \mathcal W,
		\qquad
		Y_h=h^\top Y,
		\qquad
		C_h=h^\top C,
		\eeq
		where  $
		g\in\mathbb R^{m\times r_g}$ and
		$h\in\mathbb R^{m\times r_h}$ (\(r_g, r_h\leq m\)). 
		Assume an orthogonal completion of \(S=[g\  \ g_\perp]\) satisfies
		\beq \label{eq:S_Omega_S}
		S^\top S=I_m,
		\qquad
		S^\top\Omega S=\Omega.
		\eeq
	Define matrices
		\beq \label{eq:Bga00}
		B_{g,j}
		=
		\sum_{k=1}^m g_{kj} B_k, \quad j=1,\ldots,r_g. 
		\eeq
		Then define a set
		\begin{equation}\label{eq:Rg}
			\begin{aligned}
				\mathcal R_g
				\triangleq
				\Span_{\mathbb R}
				\Bigl\{
				&M_\ell\cdots M_1B_{g,j}v:\ 
				v\in\mathbb R^s,\ j=1,\ldots,r_g,\\
				&	M_j\in\{A,B_1,\ldots,B_m\},\ 
				\ell=0,1,\ldots
				\Bigr\}.
			\end{aligned}
		\end{equation}
		($\ell=0$ means no left multiplication of $\{A,B_1,\ldots,B_m\}$ on the vector $B_{g,j}v$.) If
		\beq\label{eq:h_perp_g}
		h^\top g=0_{r_h\times r_g}
		\eeq
		and
		\beq \label{eq: ChRg}
		C_h\mathcal R_g=\{0\},
		\eeq
		then the output field observables \(Y_h\) are back-action evading with respect to the input field observables    \(\mathcal W_g\); i.e.,  \(\mathcal W_g\) have neither:
		\bed	
		\item[(1)] a direct feedthrough path to \(Y_h\); nor
		\item[(2)]  a system-variable-mediated path through the bilinear dynamics.
		\eed	
	\end{theorem}
	
	\begin{IEEEproof} Denote $\mathcal W_\perp = g_\perp^\top \mathcal W$. By \eqref{eq:S_Omega_S}, the transformed input field observables  
		\[\left[
		\begin{array}{c}
			\mathcal W_g \\
			\mathcal W_\perp
	  \end{array}\right]
	  = S^\top \mathcal W
		\]
		also satisfy \eqref{eq:W}. Thus, they are genuine quantum input quadratures. 
		Choose a full-row-rank matrix $
		N\in\mathbb R^{q\times s}$ (\(q\leq s\))
		such that $
		\ker N=\mathcal R_g$, and define a vector of observables $Z=NX$.
		By definition, \(\mathcal R_g\) is invariant under \(A\) and every \(B_\mu\), hence there exist  matrices \(\bar A\) and \(\bar B_\mu\) satisfying
		$
		NA=\bar A N,
		$
		and
	$
		NB_\mu=\bar B_\mu N $, for $ \mu=1,\ldots,m.
	$
		Equivalently, for the transformed input matrices in \eqref{eq:Bga00}, there exist matrices \(\bar B_{g,a}\) and \(\bar B_{\perp,b}\) such that
$
		NB_{g,a}=\bar B_{g,a}N,
$
		and
		$
		NB_{\perp,b}=\bar B_{\perp,b}N.
		$
		By construction of \(\mathcal R_g\) in \eqref{eq:Rg},
		$
		\operatorname{Range}(B_{g,a})
		\subseteq\mathcal R_g=\ker N.
		$
		Therefore,
	$
		NB_{g,a}=0.
	$
		Since \(N\) has full row rank, it has a right inverse \( \widetilde{N} = N^\top (NN^\top)^{-1}\) such that
	$
		N\widetilde{N}=I_q.
$
 Therefore. 
		\beq \label{eq: NB_ga3}
		\bar B_{g,a}
		=
		\bar B_{g,a}N\widetilde{N}
		=
		NB_{g,a}\widetilde{N}
		=
		0.
		\eeq
Consequently, left-multiplying  system \eqref{eq:irka_sys2} by \(N\) gives
		\beq \label{eq:Z_BAE}
		dZ
		=
		(NA_0+\bar AZ)\,dt
		+
		\sum_{b=1}^{m-r_g}
		\bar B_{\perp,b}Z\,d(\mathcal W_\perp)_b.
		\eeq
		According to \eqref{eq:WYC}, \eqref{eq:Bga00} and  \eqref{eq: NB_ga3}, the increment \(d\mathcal W_g\) does not appear  in  \eqref{eq:Z_BAE}. Next,  \eqref{eq: ChRg}
		means $
		\ker N=\mathcal R_g\subseteq\ker C_h$. 
		Hence, there exists $\bar C_h\in\mathbb R^{r_h\times q}$
		such that $C_h=\bar C_hN$. Consequently, $C_hX=\bar C_hZ$.  By \eqref{eq:WYC}, the selected output satisfies
		\[
	\begin{aligned}
			dY_h
			&=
			h^\top dY=
			C_hX\,dt+h^\top d\mathcal W
			\\
			&=
			\bar C_hZ\,dt
			+
			h^\top g\,d\mathcal W_g
			+
			h^\top g_\perp\,d\mathcal W_\perp.
			\end{aligned}
			\]
		Using \eqref{eq:h_perp_g}, we obtain
		\[
		dY_h
		=
		\bar C_hZ\,dt
		+
		h^\top g_\perp\,d\mathcal W_\perp.
		\]
		As the dynamics of  \(Z\) in \eqref{eq:Z_BAE} have no \(\mathcal W_g\)-driven term,  and the selected output $Y_h$ has no feedthrough  \(\mathcal W_g\) term either,  the output observables \(Y_h\) are back-action evading against the input observables \(\mathcal W_g\).
	\end{IEEEproof}
	
	\bmrk
	If \(Y_h\) is intended to represent a vector of simultaneously measured output fields, one should also require \(	h^\top J h=0\) for the preservation of commutation relations. 	A common normalization $
	h^\top h=I_{r_h}$ can also be imposed. 
	These conditions concern physical joint measurability of the selected output quadratures \(Y_h\) and have nothing to do with the BAE  measurement conditions presented in Theorem \ref{thm:july27_G8_general}.
	\emrk
	
	
	\section{Four-qubit-chain illustrative example}
	\label{sec:four_qubit_chain}
	
This section illustrates the results of
	Sections~\ref{sec:structural decomp}--\ref{sec:BAE} using a four-qubit
	Heisenberg chain with collective dephasing. The example exhibits three
	coexisting structures: a logical-qubit observable subsystem (DFS),
	a conserved commutative magnetization algebra (QND observables), and a homodyne output
	that evades the conjugate measurement's back action (BAE measurement). Four-qubit
	decoherence-free encodings and their control by exchange interactions
	are established construction; see for example 
\cite{BaconKempeLidarWhaley00,LidarBaconKempeWhaley00,AntonioBose13}.

	\subsection{The model}
	\label{subsec:four_chain_model}
	
	Let $\mathcal H=(\mathbb C^2)^{\otimes4}$. Then $n=2^4=16$ and $s=n^2-1=255$.
	For $a\in\{x,y,z\}$, let $\sigma_a^{(j)}$ denote the Pauli matrix
	$\sigma_a$ acting on qubit $j$. Define the nearest-neighbor exchange
	operators $E_{j,j+1}$ and chain Hamiltonian $H_{\mathrm{ch}}$ as
	\begin{equation*}
		\begin{aligned}
			E_{j,j+1}
			&=\sum_{a=x,y,z}\sigma_a^{(j)}\sigma_a^{(j+1)},
			\qquad j=1,2,3,\\
			H_{\mathrm{ch}}
			&=J_{12}E_{12}+J_{23}E_{23}+J_{34}E_{34},
		\end{aligned}
		\label{eq:four_chain_H}
	\end{equation*}
	where $J_{12},J_{23},J_{34}\in\mathbb R$.
	Introduce the collective spin components
	\begin{equation}
		J_a \triangleq \frac12\sum_{j=1}^4\sigma_a^{(j)}, \quad  
		\mathbf{J}^2 \triangleq J_x^2+J_y^2+J_z^2,  \quad a=x,y,z.
		\label{eq:four_chain_Jz}
	\end{equation}
It can be readily shown that   \([\sigma_a^{(j)}+\sigma_a^{(j+1)} ,E_{j,j+1}]=0\) and \([\sigma_a^{(k)},E_{j,j+1}]=0\) if \(k\notin\{j,j+1\}\). Hence we have
	\begin{equation}
		[H_{\mathrm{ch}},J_a]=0,   \quad   
		[H_{\mathrm{ch}},\mathbf{J}^2]=0, \quad a=x,y,z.
		\label{eq:four_chain_JH_commute}
	\end{equation}
	Assume a single vacuum field channel with coupling operator 
	$L=\sqrt{\kappa}\,J_z$, where $\kappa>0$.
	Let $c_z\in\mathbb R^{255}$ satisfy $\Phi(c_z)=J_z$. Then by \eqref{eq:def_Phi} we get 
	\beq \label{eq:L_example}
	L=\Phi(\sqrt{\kappa}c_z)J_z = (\sqrt{\kappa}c_z)^\top X.
	\eeq
	Consequently,  system matrices in
	\eqref{eq:N-from-C} are
	\begin{equation}
		\begin{aligned}
			C_1&=2\sqrt{\kappa}\,c_z^\top, & C_2&=0,   &A_0&=0, \\
			B_1&=0,                                             & B_2&=-2\sqrt{\kappa}\,\Theta_n^-(c_z), 
		\end{aligned}
		\label{eq:four_chain_coefficients}
	\end{equation}
and the input--output equations in \eqref{eq:irka_sys2} are
	\begin{equation} 	\label{eq:four_chain_output}
			dY_1=2\sqrt{\kappa}\,J_z\,dt+dW_1, \quad
			dY_2=dW_2.
	\end{equation}
	The $Y_1$-quadrature carries the magnetization signal $J_z$, whereas
	the conjugate $Y_2$-quadrature has no system-dependent signal.

	For $H_{\mathrm{ch}}$ an $L$ given above,  in the Heisenberg picture, the evolution of  the system operator $O$ follows the QSDE, \cite[Section 2]{Duffaut-et-al_2016}, 
	\beq \label{eq:QSDE}
	d O = \mathcal L(O) dt +\ii [L,O] dW_2,
	\eeq 
	where the superoperators
	\begin{equation}
		\begin{aligned}
			\mathcal L(O)&=\mathcal K(O)+\mathcal D_L(O),\\
			\mathcal K(O)&=-\ii[O, H_{\mathrm{ch}}], \quad
			\mathcal D_L(O)&=-\frac{1}{2}[L,[L,O]].
		\end{aligned}
		\label{eq:four_chain_generator}
	\end{equation}
	Clearly 
	\beq \label{eq:Z_conserved}
	dJ_z =0,
	\eeq 
i.e., the measured magnetization $J_z$ is conserved.  Moreover, from \eqref{eq:L_example} we have $c_z^\top dX=0$. Comparing with  \eqref{sys_n_level} and using \eqref{eq:four_chain_coefficients}, we get
		\begin{equation}
		c_z^\top A=0,
		\qquad
		c_z^\top B_\mu=0, 
		\quad \mu=1,2.
		\label{eq:four_chain_conservation}
	\end{equation}
Similarly, given $v\in\mathbb R^{255}$,  by means of   $B_2=-B_2^\top$,  \eqref{eq:four_chain_coefficients}  and \eqref{eq:isometry_4sep},  substituting $O=\Phi(v) =v^\top X$ into \eqref{eq:QSDE} and comparing with  \eqref{sys_n_level}   we obtain 
	\begin{equation}
		\begin{aligned}
			\mathcal L(\Phi(v))
			&=\Phi(A^\top v),\\
			\Phi(B_2v)
			&=-\ii\sqrt{\kappa}\,[J_z,\Phi(v)],\\
			C_1v
			&=\sqrt{\kappa}\,\Tr(J_z\Phi(v)).
		\end{aligned}
		\label{eq:four_chain_dictionary}
	\end{equation}
Hence, due to the isometric property of $\Phi$ defined in \eqref{eq:def_Phi},
\beq \label{eq:10sep_temp1}
B_2v=0 \Longleftrightarrow [J_z,\Phi(v)]=0.
\eeq

	\subsection{An autonomous logical-qubit observable subsystem (DFS)}
	\label{subsec:four_chain_logical}
	
	\paragraph{The singlet encoding}
	Define the two-qubit states
	\begin{equation}
		\begin{aligned}
			|s\rangle&=\frac{|01\rangle-|10\rangle}{\sqrt2},
			&
			|t_0\rangle&=\frac{|01\rangle+|10\rangle}{\sqrt2},\\
			|t_+\rangle&=|00\rangle,
			&
			|t_-\rangle&=|11\rangle.
		\end{aligned}
		\label{eq:four_chain_pair_states}
	\end{equation}
Then define orthonormal states
	\begin{equation}
		\begin{aligned}
			|0_L\rangle
			&=|s\rangle_{12}|s\rangle_{34},\\
			|1_L\rangle
			&=\frac1{\sqrt3}\bigl(
			|t_+\rangle_{12}|t_-\rangle_{34}
			-|t_0\rangle_{12}|t_0\rangle_{34}
			+|t_-\rangle_{12}|t_+\rangle_{34}
			\bigr).
		\end{aligned}
		\label{eq:four_chain_logical_states}
	\end{equation}
	It can be easily verified that 
	\beq \label{eqLJ_0_1}
	J_a |0_L\rangle =  J_a |1_L\rangle =0, \quad  a=x,y,z.
	\eeq
	 As a result, by \eqref{eq:four_chain_Jz}, $\mathbf{J}^2|0_L\rangle= \mathbf{J}^2|1_L\rangle = 0$. Hence, $\{\ket{0_L}, \ket{1_L}\}$ form an orthonormal basis for the two-dimensional total-spin-zero subspace. Moreover,  the rank-2
	projector
	\begin{equation}
		P_0=|0_L\rangle\langle0_L|+|1_L\rangle\langle1_L|
		\label{eq:four_chain_P0}
	\end{equation}
	is the spectral projector of $\mathbf{J}^2$ associated with eigenvalue zero.  Hence the second
	equation in \eqref{eq:four_chain_JH_commute} implies
	\begin{equation}
		[H_{\mathrm{ch}},P_0]=0.
		\label{eq:four_chain_dark_conditions}
	\end{equation}
	Moreover, by \eqref{eqLJ_0_1}--\eqref{eq:four_chain_P0},  
	\[
	J_aP_0=P_0J_a=0, \quad  a=x,y,z.
	\]
	Thus the logical subspace spanned by \(\{|0_L\rangle, |1_L\rangle\}\) is invariant under the exchange Hamiltonian $H_{\mathrm{ch}}$
	and dark with respect to collective dephasing $J_z$.
	
	Define logical Pauli operators on this subspace, extended by zero
	on its orthogonal complement:
	\begin{equation}
		\begin{aligned}
			\bar\sigma_x
			&=|0_L\rangle\langle1_L|+|1_L\rangle\langle0_L|,\\
			\bar\sigma_y
			&=-\ii|0_L\rangle\langle1_L|+\ii|1_L\rangle\langle0_L|,\\
			\bar\sigma_z
			&=|0_L\rangle\langle0_L|-|1_L\rangle\langle1_L|.
		\end{aligned}
		\label{eq:four_chain_logical_Pauli}
	\end{equation}
By \eqref{eqLJ_0_1}--\eqref{eq:four_chain_P0}, 
	\begin{equation}
		\begin{aligned}
			&P_0\bar\sigma_a=\bar\sigma_aP_0=\bar\sigma_a, \ \ J_z\bar\sigma_a=\bar\sigma_aJ_z=0, \\
			&\bar\sigma_a\bar\sigma_b
			=\delta_{ab}P_0
			+\ii\sum_{c=x,y,z}\epsilon_{abc}\bar\sigma_c,
		\end{aligned}
		\label{eq:four_chain_logical_product}
	\end{equation}
		where \(\epsilon_{abc}\) is the completely antisymmetric Levi--Civita
	tensor \cite{Duffaut-et-al_2012}. 
	Direct, yet tedious, calculation on this logical basis gives
	\begin{equation}
		\begin{aligned}
			P_0E_{12}P_0=P_0E_{34}P_0
			&=-P_0-2\bar\sigma_z,\\
			P_0E_{23}P_0
			&=-P_0+\bar\sigma_z+\sqrt3\,\bar\sigma_x.
		\end{aligned}
		\label{eq:four_chain_exchange_restriction}
	\end{equation}
	Consequently, the restricted Hamiltonian is
	\begin{equation}
		\begin{aligned}
			P_0H_{\mathrm{ch}}P_0
			&=h_0P_0+h_x\bar\sigma_x+h_z\bar\sigma_z,\\
			h_0&=-(J_{12}+J_{23}+J_{34}),\\
			h_x&=\sqrt3\,J_{23},\\
			h_z&=J_{23}-2J_{12}-2J_{34}.
		\end{aligned}
		\label{eq:four_chain_H_logical}
	\end{equation}

	\paragraph{Verification of strong alignment}
	Denote 
	\[
	Q_L=P_0-\frac18I_{16}.
	\] 
	Then  by \eqref{eq:four_chain_logical_product} we get
		\begin{equation}\label{eq:ccr_anticcr}
		\begin{aligned}
			\Tr(\bar\sigma_a\bar\sigma_b) &= 2\delta_{ab},
			\\
			[\bar\sigma_a,\bar\sigma_b]
			&=
			2\ii\sum_{c=x,y,z}\epsilon_{abc}\bar\sigma_c,
			\\
			\{\bar\sigma_a,\bar\sigma_b\}
			&=
			2\delta_{ab}P_0
			=
			\frac14\delta_{ab}I_{16}+ 2\delta_{ab}Q_L.
		\end{aligned}
		\end{equation} 
Define real coordinate subspaces $ \mathcal S, \mathcal Q\subset \mathbb{R}^{255}$ by
		\begin{equation} 	\label{eq:four_chain_S}
				\Phi(\mathcal S)
				=\Span_{\mathbb R}
				\{\bar\sigma_x,\bar\sigma_y,\bar\sigma_z\}, \quad 
				\Phi(\mathcal Q)=\Span_{\mathbb R}\{Q_L\}.
		\end{equation}
	Let $\Phi(s_a)=\bar\sigma_a$ for $a=x,y,z$, and $\Phi(q_L)=Q_L$.
	Since $\Tr(P_0)=2$, $Q_L$ is traceless Hermitian. By \eqref{eq:isometry_4sep} and \eqref{eq:ccr_anticcr},  normalized Hilbert--Schmidt inner products give
	\begin{equation}
		s_a^\top s_b=\delta_{ab},
		\qquad
		q_L^\top s_a=0,
		\qquad
		q_L^\top q_L=\frac78.
		\label{eq:four_chain_coords}
	\end{equation}
	Consequently, $\mathcal Q\subseteq\mathcal S^\perp$.  Define the projector 
	\beq \label{eq:P_S_example}
	P_{\mathcal S}=\sum_{a=x,y,z} s_as_a^\top
	\eeq 
	on the subspace \(\mathcal  Q\). By \eqref{eq:Phi_uv_ccr_anticcr},  \eqref{eq:ccr_anticcr}, and the facts
	$Q_L\bar\sigma_a=\frac78\bar\sigma_a$ and
	$Q_L^2=\frac34Q_L+\frac7{64}I_{16}$,  we obtain
	\begin{equation}
		\begin{aligned}
			&f(s_a,s_b)=-\sum_c\epsilon_{abc}s_c,\ \ 
			d(s_a,s_b)=\delta_{ab}q_L,\\
			& f(q_L,\mathcal S\oplus\mathcal Q)=\{0\},\ 
			d(q_L,s_a)=\frac78s_a, \ 
			d(q_L,q_L)=\frac34q_L.
		\end{aligned}
		\label{eq:four_chain_alignment_products}
	\end{equation}
	These identities verify  \eqref{eq:f_d_Q_S} in
	Definition~\ref{def:strong_alignment}.

By  \eqref{eq:four_chain_S}, the associated unital $C^*$-algebra in \eqref{eq:july25_2.5}  is
	\begin{equation}
		\begin{aligned}
			\mathfrak A
			&=\Span_{\mathbb C}
			\{I_{16},Q_L,\bar\sigma_x,\bar\sigma_y,\bar\sigma_z\}\\
			&=P_0M_{16}(\mathbb C)P_0
			\oplus\mathbb C(I_{16}-P_0),
		\end{aligned}
		\label{eq:four_chain_logical_algebra_0}
	\end{equation}
	where the facts $	\Span_{\mathbb C}\{I_{16},Q_L\}
	=
	\Span_{\mathbb C}\{P_0,I_{16}-P_0\}$ and $P_0M_{16}(\mathbb C)P_0
	=
	\Span_{\mathbb C}
	\{P_0,\bar\sigma_x,\bar\sigma_y,\bar\sigma_z\}$ have been used to establish the second equality.  Clearly,  the center
	$Z(\mathfrak A)=\Span_{\mathbb C}\{I_{16},Q_L\}$.
	Using \eqref{eq:four_chain_coords} and the traceless property of the logical
	Pauli operators, we obtain
	\begin{equation}
			Z(\mathfrak A)_{\mathrm{sa}}\cap\{I_{16}\}^{\perp}
			=\Phi(\mathcal Q), \ \ 
			\mathfrak A_{\mathrm{sa}}
			\cap Z(\mathfrak A)_{\mathrm{sa}}^\perp
			=\Phi(\mathcal S).
		\label{eq:four_chain_center}
	\end{equation}
	As a result, the required exact central--noncentral separation \eqref{eq:Phi_QS} also holds,
	and $\mathcal S$ is strongly aligned.
	
	\paragraph{Verification of \eqref{eq:zero_pattern}} By \eqref{eq:four_chain_logical_product}, we get $[J_z,\Phi(s_a)]=0$. Then by \eqref{eq:10sep_temp1} and \eqref{eq:four_chain_coefficients}, we have $B_\mu s_a=0$. Moreover, by \eqref{eq:four_chain_logical_product}, $c_z\perp s_a$ and hence $C_1s_a=0$. Finally, $C_2=0$ and $A_0=0$ are already given in \eqref{eq:four_chain_coefficients}. Consequently, \eqref{eq:july25_5.5}-\eqref{eq:july25_5.7} hold.
	
	Next, we look at the drift condition \eqref{eq:july25_5.4}.  from \eqref{eq:four_chain_dark_conditions},  \eqref{eq:four_chain_logical_product},  \eqref{eq:four_chain_H_logical}, and  \eqref{eq:four_chain_generator}, we get
	\begin{equation}
		\begin{aligned}
			\mathcal K(\bar\sigma_a) &= \mathcal K(P_0\bar\sigma_a P_0)
			=\ii[h_x\bar\sigma_x+h_z\bar\sigma_z,\bar\sigma_a]
			\in\Phi(\mathcal S),\\
			\mathcal D_L(\bar\sigma_a)&  =0.
		\end{aligned}
		\label{eq:four_chain_G_on_S}
	\end{equation}
Hence
	\(\mathcal L\bigl(\Phi(\mathcal S)\bigr)\subseteq\Phi(\mathcal S)\).  For \(X,Y\in\Hermzero(16)\),
	\begin{align*}
		&\langle X,\mathcal K(Y)\rangle_{\rm nHS}
		=
		-\frac{\ii}{2}\Tr\bigl(XYH_{\mathrm{ch}}-XH_{\mathrm{ch}}Y\bigr)
		\\
		=&
		\frac12\Tr\bigl(\ii(XH_{\mathrm{ch}}-H_{\mathrm{ch}}X)Y\bigr)
		=
		-\langle\mathcal K(X),Y\rangle_{\rm nHS}.
	\end{align*}
	That is, the Hamiltonian superoperator \(\mathcal K\) is
	skew-adjoint with respect to the normalized Hilbert--Schmidt inner
	product \eqref{eq:nHS}. 
	Similarly, from  $
	\langle X,\mathcal D_L(Y)\rangle_{\rm nHS}
	=
	\langle\mathcal D_L(X),Y\rangle_{\rm nHS}$ we know that the Lindblad  superoperator \(\mathcal D_L\) is self-adjoint.  
	Now let \(X\in\Phi(\mathcal S)\) and \(Y\in\Phi(\mathcal S)^\perp\),
	the orthogonal complement being taken in \(\Hermzero(16)\).  Then by \eqref{eq:four_chain_G_on_S}
	\[
	\langle X,\mathcal L(Y)\rangle_{\rm nHS}
	=
	-\langle\mathcal K(X),Y\rangle_{\rm nHS}
	+
	\langle\mathcal D_L(X),Y\rangle_{\rm nHS}
	=
	0.
	\]
	Thus
	\(\mathcal L\bigl(\Phi(\mathcal S)^\perp\bigr)\subseteq\Phi(\mathcal S)^\perp\).
	Since \(\Phi\) is an isometric isomorphism,
	\(\Phi(\mathcal S)^\perp=\Phi(\mathcal S^\perp)\), and by \eqref{eq:four_chain_dictionary} and  
	\eqref{eq:four_chain_G_on_S} we get $
	A^\top\mathcal S\subseteq\mathcal S$ and  $
	A^\top\mathcal S^\perp\subseteq\mathcal S^\perp$;
	i.e., \([A^\top,P_{\mathcal S}]=0\), which is equivalent to
	\([A,P_{\mathcal S}]=0\).  This is \eqref{eq:july25_5.4}. Collecting the preceding conclusions yields precisely \eqref{eq:zero_pattern}. Thus, all hypotheses of Theorem~\ref{thm:decomp} are satisfied.

	\paragraph{Artin--Wedderburn decomposition and DFS}
	Choose $U\in SU(16)$ such that
	$U|1_L\rangle=e_1$, $U|0_L\rangle=e_2$, and the orthogonal complement  \(\Range(I_{16}-P_0)\) 
	is mapped onto $\Span\{e_3,\ldots,e_{16}\}$.
	With the sign convention for the generalized Gell--Mann basis  used in
	Section~\ref{sec:preli}, we have
	\begin{equation}
		\begin{gathered}
			U\bar\sigma_xU^\dagger=u_{12},  \ 
			U\bar\sigma_yU^\dagger=v_{12},  \ 
			U\bar\sigma_zU^\dagger=w_1, 
			\\
			UP_0U^\dagger=I_2\oplus0_{14},  \ 
			UQ_L U^\dagger=\frac78I_2\oplus\Bigl(-\frac18\Bigr)I_{14}.
		\end{gathered}
		\label{eq:four_chain_U_action}
	\end{equation}
	Consequently, applying  unitary conjugation to $\mathfrak A$ in the second equality in  \eqref{eq:four_chain_logical_algebra_0} and $\Phi(\mathcal S)$ and  $\Phi(\mathcal Q)$ in   \eqref{eq:four_chain_S} and using \eqref{eq: Z_0 and D}, we get
	\begin{equation}
		\begin{aligned}
			U\mathfrak A U^\dagger
			&=M_2(\mathbb C)\oplus\mathbb CI_{14},\\
			\mathfrak D
			&=\Hermzero(2)\oplus0_{14}=\Span_{\mathbb R}\{u_{12},v_{12},w_1\},\\
			\mathfrak Z_0
			&=\Span_{\mathbb R}\{7I_2\oplus(-I_{14})\}.
		\end{aligned}
		\label{eq:four_chain_D_Z0}
	\end{equation}
	The Artin--Wedderburn data are
	$(r_1,m_1)=(2,1)$ and $(r_2,m_2)=(1,14)$. Moreover, $\dim\mathfrak D=3$ and $\dim\mathfrak Z_0=1$,
	corresponding exactly to $\mathcal S$ and $\mathcal Q$.
	
	Relabel the original Gell--Mann matrices so that
	$\lambda_1=u_{12}$, $\lambda_2=v_{12}$, and $\lambda_3=w_1$. Then by \eqref{eq:four_chain_D_Z0}, \(\mathfrak D=\Span_{\mathbb R}\{\lambda_1,\lambda_2,\lambda_3\}\).
	Thus, the unique noncommutative block has multiplicity $m_1=1$ and occupies
	the leading corner, so the conditions of
	Corollary~\ref{cor:fixed_frame_decomp} hold.
	Writing
	$\bar{\bm{\sigma}}
	=(\bar\sigma_x,\bar\sigma_y,\bar\sigma_z)^\top$, direct calculation by \eqref{eq:four_chain_G_on_S} and \eqref{eq:ccr_anticcr} gives the DFS:
	\begin{equation}
		\begin{aligned}
			d\bar{\bm{\sigma}}
			&=A_{\mathcal S}\bar{\bm{\sigma}}\,dt,\\
			A_{\mathcal S}
			&=T_{\mathcal S}^\top AT_{\mathcal S}
			=
			\begin{bmatrix}
				0&-2h_z&0\\
				2h_z&0&-2h_x\\
				0&2h_x&0
			\end{bmatrix}.
		\end{aligned}
		\label{eq:four_chain_logical_dynamics}
	\end{equation}
	On the encoded subsystem \(\bar{\bm{\sigma}}\), the Bloch vector precesses about
	$(h_x,0,h_z)$ with angular frequency
	$2\sqrt{h_x^2+h_z^2}$.  We have shown that the  transformed system has  a DFS of the form in  \eqref{eq:block_sector_dynamics}. We verified 	Corollary~\ref{cor:fixed_frame_decomp}.

	\subsection{The conserved QND magnetization algebra}
	\label{subsec:four_chain_QND}
	
	The spectra of $J_z$ are $\{-2,-1,0,1,2\}$.
	Let $\Pi_m$ denote the spectral projector corresponding to  the spectrum $m$.
Calculations show that  the rank of $\Pi_m$   is
	\begin{equation}
		d_m=\Tr(\Pi_m)=\binom4{2-m},
		\quad
		 m= -2,-1,0,1,2.
		\label{eq:four_chain_dm}
	\end{equation}
	Define traceless observables
		\begin{equation}\label{eq:Q_m_23Sep}
	Q_m
	\triangleq
	\Pi_m-\frac{d_m}{16}I_{16}, \quad m= -2,-1,0,1,2, 
	\end{equation} 
	and their coordinate vectors by $\Phi(v_m) = Q_m$.
	Denote
	\begin{equation}
		\begin{aligned}
			V&=[\,v_{-2}\ v_{-1}\ v_0\ v_1\,] \in \mathbb{R}^{255\times 4},\\
			Z&=V^\top X=[Q_{-2},Q_{-1},Q_0,Q_1]^\top.
		\end{aligned}
		\label{eq:four_chain_V}
	\end{equation}
As  $\Pi_m$ are  mutually orthogonal, simple algebraic manipulations yield that  $Q_{-2},Q_{-1},Q_0, Q_1$ are linearly independent. As a result, the matrix $V$ has full column rank.
	
	\paragraph{Commutativity and conservation}
	As spectral projectors $\Pi_m$ are mutually orthogonal,  $Q_{-2},Q_{-1},Q_0, Q_1$ mutually 
	commute too. From \eqref{eq:Q_m_23Sep} we have
	\begin{equation}
		\begin{aligned}
			Q_mQ_\ell
			&=\left(\delta_{m\ell}-\frac{d_\ell}{16}\right)Q_m
			-\frac{d_m}{16}Q_\ell\\
			&\quad
			+\frac{d_m}{16}
			\left(\delta_{m\ell}-\frac{d_\ell}{16}\right)I_{16}.
		\end{aligned}
		\label{eq:four_chain_Q_product}
	\end{equation}
	Moreover, by \eqref{eq:isometry_4sep}, 
	$v_m^\top v_\ell
	=\frac{d_m}{2}(\delta_{m\ell}-d_\ell/16)$. Consequently, comparing \eqref{eq:four_chain_Q_product} with the product formula
	\eqref{eq:july25_1.3} yields
	\begin{equation}
		\begin{aligned}
			f(v_m,v_\ell)&=0,\\
			d(v_m,v_\ell)
			&=\left(\delta_{m\ell}-\frac{d_\ell}{16}\right)v_m
			-\frac{d_m}{16}v_\ell
			\in\Range(V).
		\end{aligned}
		\label{eq:four_chain_QND_closure}
	\end{equation}
	These are the hypotheses of
	Theorem~\ref{thm:july27_qnd_G7}(i)--(ii). Clearly, 
	\begin{equation}
		\begin{aligned}
			\mathfrak A_{Z}
			&=\Span_{\mathbb C}\{I_{16},Q_{-2},Q_{-1},Q_0,Q_1\}\\
			&=\Span_{\mathbb C}\{\Pi_{-2},\Pi_{-1},\Pi_0,\Pi_1,\Pi_2\}
		\end{aligned}
		\label{eq:four_chain_Jz_algebra}
	\end{equation}
	is a five-dimensional commutative unital $C^*$-algebra.
	A set of its minimal projectors is precisely  \(\{\Pi_{-2},\Pi_{-1},\Pi_0,\Pi_1,\Pi_2\}\). This verifies \eqref{eq:july27_qnd_3.3}.
	
	Because $J_z$  commutes with
	$H_{\mathrm{ch}}$ and $L$,  each $Q_m$ defined in \eqref{eq:Q_m_23Sep} commutes with
	$H_{\mathrm{ch}}$ and $L$ too.  As a result,
\eqref{eq:four_chain_generator}
	and \eqref{eq:four_chain_dictionary} give
	\begin{equation}
		V^\top A=0,
		\qquad
		V^\top B_\mu=0,
		\qquad
		V^\top A_0=0.
		\label{eq:four_chain_QND_conditions}
	\end{equation}
The conditions in
	Theorem~\ref{thm:july27_qnd_G7}(iii) therefore hold with $K=0$ and
	$b=0$. Consequently  \eqref{eq:july27_qnd_9.5} becomes $
		Z(t)\equiv Z(0)$, and \eqref{eq:july27_qnd_9.6} holds. 
	Thus, the complete magnetization QND algebra $\mathfrak A_{Z}$ is conserved,	not merely its generator $J_z$.
	
	\paragraph{Linear observability}
	Choose $M=[\,1\ \  0\,]$.  Then $Y_M=Y_1$ and $MJM^\top=0$.
	Using 
	\[
	J_z =\sum_m m \Pi_m, \quad \sum_m Q_m=0,
	\]
	 and \eqref{eq:four_chain_dm}-\eqref{eq:four_chain_V}, we get
	\beq
			J_z
			=-4Q_{-2}-3Q_{-1}-2Q_0-Q_1=\gamma^\top Z,
		\label{eq:four_chain_Jz_Q}
\eeq
	where \(\gamma=[-4 \ -3   ~ \ -2 \ -1]^\top\). 
	Hence $c_z=V\gamma$ and
	$MC=C_1=2\sqrt{\kappa}\,\gamma^\top V^\top$.
	This verifies \eqref{eq:KerV}, as required in
	Theorem~\ref{thm:july27_qnd_G7}(iv).  Also,  \eqref{eq:Y_M}--\eqref{eq:C_Z} become
	\begin{equation}		\label{eq:four_chain_CZ}
			dY_1=C_ZZ\,dt+dW_1, 
			\quad
			C_Z=2\sqrt{\kappa}\gamma^\top.
	\end{equation}
	Since $K=0$,  by \eqref{eq:O_Z} we  get \(\rank\mathcal O_Z=\rank C_Z=1<4\). Hence, , 
	the vector \(Z\) is not observable from \(Y_1\) in the
standard linear-systems sense.

	\paragraph{Spectral separation} 	
	Let $F=C_ZZ=2\sqrt{\kappa}\,J_z$, where \eqref{eq:four_chain_Jz_Q}--\eqref{eq:four_chain_CZ} have been used. Clearly, $F$ has five distinct
	spectral values $f_m=2\sqrt{\kappa}\,m$. \eqref{eq:nonparallel} holds naturally. 
	The spectral decomposition of $F$ is
	\begin{equation}
		F=\sum_{k=-2}^{2}f_k\Pi_k, \quad f_k = 2\sqrt{\kappa} k.
		\label{eq:F-spectral-decomposition}
	\end{equation}
Define degree-4  polynomials 
\beq
	p_m(x)
	\triangleq
	\prod_{\substack{\ell=-2\\\ell\neq m}}^{2}
	\frac{x-f_\ell}{f_m-f_\ell}, \quad m=-2,\ldots,2.
			\label{eq:p_polynomial_m}
\eeq
	Clearly,  $p_m(f_k)=\delta_{mk}$.
	To evaluate these polynomials at the matrix $F$, we rewrite $p_m(x)$ as 
		\begin{equation}
	p_m(x)=\sum_{r=0}^{4}a_r^{(m)}x^r.
		\label{eq:p_polynomial}
\end{equation}
On the other hand, by \eqref{eq:F-spectral-decomposition},
	\begin{equation}
		F^r=\sum_{k=-2}^{2}f_k^r\Pi_k,
		\qquad r\in \mathbb{Z}.
		\label{eq:powers-of-F}
	\end{equation}
	Using \eqref{eq:p_polynomial}--\eqref{eq:powers-of-F}, we obtain
	\begin{align}
		&p_m(F)=\sum_{r=0}^{4}a_r^{(m)}F^r
		=\sum_{r=0}^{4}a_r^{(m)}
		\sum_{k=-2}^{2}f_k^r\Pi_k
		\nonumber\\
		=&\sum_{k=-2}^{2}
		\left(
		\sum_{r=0}^{4}a_r^{(m)}f_k^r
		\right)\Pi_k
		=\sum_{k=-2}^{2}p_m(f_k)\Pi_k
		\nonumber\\
		=&\sum_{k=-2}^{2}\delta_{mk}\Pi_k
		=\Pi_m.
		\label{eq:Lagrange-projection}
	\end{align}
	By \eqref{eq:p_polynomial_m}, \eqref{eq:Lagrange-projection} and the identity \(I_{16}=\sum_{m=-2}^2 \Pi_m\) we get
\[
		\Pi_m
		=p_m(F)
 =\prod_{\substack{\ell=-2\\ \ell\neq m}}^{2}
		\frac{F-f_\ell I_{16}}{f_m-f_\ell}
	=\prod_{\substack{\ell=-2\\ \ell\neq m}}^{2}
		\frac{F-2\sqrt{\kappa}\,\ell I_{16}}
		{2\sqrt{\kappa}(m-\ell)}
	\]
 for all $ m=-2,\ldots,2$. 
	Therefore, every  projector \(\Pi_m\) is a polynomial in $F$ of degree
	at most four, and consequently
	$\mathfrak A_{Z}
	=\Span_{\mathbb C}\{I_{16},F,F^2,F^3,F^4\}$. Equivalently, the output-signal observable \(F\) separates spectral projections of
	 $\mathfrak A_{Z}$, and  $\alg^*(I_{16},F)=\mathfrak A_{Z}$, which is \eqref{eq:alg_A_Z}.  Theorem \ref{thm:QND_2} is verified.
	
The above distinction is essential: linear  observability
	concerns a linear model composed of \eqref{eq:july27_qnd_9.5} and \eqref{eq:Y_M}, whereas algebraic observability
	uses the finite joint spectrum of the commuting observables.
	The Homodyne signal distinguishes the five magnetization labels.

	\subsection{Quantum BAE measurements}
	\label{subsec:four_chain_BAE}
	
	For the single field channel, let
	\begin{equation*}
		h=e_1, \ g=e_2, \ 
		g_\perp=-e_1, \ 
		S=[\,g\ g_\perp\,]
		=
		\begin{bmatrix}
			0&-1\\
			1&0
		\end{bmatrix}.
		\label{eq:four_chain_gh}
	\end{equation*}
	Then $W_g=W_2$, $Y_h=Y_1$, $B_g=B_2$, $C_h=C_1$. 
	Moreover, $h^\top g=0$ excludes direct feedthrough from $W_2$ to
	$Y_1$. Clearly, \eqref{eq:S_Omega_S},  and both
	$h^\top Jh=0$ and $h^\top h=1$ also hold.
	
	It remains to verify the system-mediated decoupling condition. By \eqref{eq:four_chain_conservation},  we get $C_hA=0$,  $C_hB_1=0$, $C_hB_2=0$; 
	that is, \(C_h\) annihilates every matrix in \(\{A,B_1,B_2\}\) from
	the left.  Let \(M_j\in\{A,B_1,B_2\}\).  It follows that  for every
	\(v\in\mathbb R^{255}\),
\[
		C_hM_\ell\cdots M_1B_gv=0,
		\qquad
		\ell=0,1,\ldots.
		\label{eq:four_chain_BAE_word}
\]
Thus \eqref{eq: ChRg} holds. 	All conditions of
	Theorem~\ref{thm:july27_G8_general} are therefore satisfied. Actually, as  already derived in \eqref{eq:four_chain_output},  the output quadrature 
	\begin{equation}
		dY_h
		=
		2\sqrt{\kappa}\,J_z\,dt+dW_1,
		\label{eq:four_chain_BAE_output}
	\end{equation}
	is back-action evading with
	respect to the input quadrature \(W_g=W_2\). Theorem~\ref{thm:july27_G8_general}  is verified. Moreover, by \eqref{eq:Z_conserved},  $J_z(t) \equiv J_z(0)$, the output equation \eqref{eq:four_chain_BAE_output} takes the particularly simple
	form $Y_1(t)=2\sqrt{\kappa}\,J_z(0)t+W_1(t)$.
	Thus the input field  quadrature $W_2$ has neither a direct nor a system-mediated effect on  the output quadrature $Y_1$. We verified Theorem~\ref{thm:july27_G8_general}.

	\subsection{Coexistence of the three structures}
The three structures exhibited above are governed by two commutation
properties of the model: the rotation invariance
\eqref{eq:four_chain_JH_commute} of the exchange Hamiltonian  \(H_{\mathrm{ch}}\), and the
fact that the coupling operator $L$ is a function of the conserved
magnetization $J_z$.  The first property makes the total-spin-zero subspace
invariant under \(H_{\mathrm{ch}}\) and, together with
\(J_zP_0=0\), yields the autonomous 
logical qubit in Subsection~\ref{subsec:four_chain_logical}, which is
the \((r,m)=(2,1)\) DFS.  These two properties together yield  both
the QND magnetization algebra of
Subsection~\ref{subsec:four_chain_QND}, which is conserved and, although
not linearly observable, is generated by the single homodyne signal $F$, 
and the quantum BAE measurement of
Subsection~\ref{subsec:four_chain_BAE}.

\section{Conclusion}
We have investigated the structure of quantum finite-level systems in the Heisenberg picture,  by focusing on their decoherence-free subsystems (DFSs), quantum nondemolition (QND) variables and quantum back-action evading (BAE) measurements. One line of future research is to apply these analytical results to system design of quantum finite-level systems. The Tavis–Cummings model is studied \cite{DZWW23}, where for a special case investigated in Section  III-C, a real orthogonal coordinate transformation \(T\)  decomposes the whole system into the controllable/observable subsystem and a uncontrollable/unobservable subsystem (namely, DFS).  Interestingly, the columns of \(T\) corresponding to the DFS are exactly the dark states which are decoherence-free subspace in the Schrödinger picture.  However, to the best of our knowledge,  the connection between decoherence-free subsystems  in the Heisenberg picture  and decoherence-free subspaces in the Schrödinger picture does not exist for general quantum systems.  Thus,  anther line of future research is to connect these results in the Heisenberg picture to relevant results in the Schrödinger picture, in the hope of  obtaining a unified theory.

\end{document}